\documentclass[11pt,a4paper]{article}
\usepackage[T1]{fontenc}
\usepackage[utf8]{inputenc}
\usepackage{lmodern,microtype}
\usepackage{authblk}
\usepackage{capt-of}
\usepackage{tikz}
\usetikzlibrary{arrows.meta,positioning,shapes.geometric}
\usepackage[a4paper,margin=27mm]{geometry}
\usepackage{amsmath,amssymb,amsthm,mathtools,mathrsfs,bm}
\usepackage{booktabs,tabularx,longtable,array,enumitem,aliascnt,xcolor}
\usepackage[colorlinks=true,linkcolor=blue!55!black,citecolor=blue!55!black,urlcolor=blue!55!black]{hyperref}
\usepackage[nameinlink,noabbrev]{cleveref}
\hypersetup{pdftitle={Closed-Form Arbitrary-Gap Rényi Nonadditivity Above Any Positive Order via Direct Products of Free Groups and High-Girth Graphs},pdfauthor={Guocheng Zhen, Ranyiliu Chen, Xin Wang},pdfsubject={Arbitrary positive Renyi gaps, high-girth finite models, analytic filtering, and dimension bounds}}
\setlist{nosep}
\allowdisplaybreaks[1]
\newtheorem{theorem}{Theorem}[section]
\newaliascnt{lemma}{theorem}
\newtheorem{lemma}[lemma]{Lemma}\aliascntresetthe{lemma}
\newaliascnt{proposition}{theorem}
\newtheorem{proposition}[proposition]{Proposition}\aliascntresetthe{proposition}
\newaliascnt{corollary}{theorem}
\newtheorem{corollary}[corollary]{Corollary}\aliascntresetthe{corollary}
\theoremstyle{definition}
\newaliascnt{definition}{theorem}
\aliascntresetthe{definition}
\theoremstyle{remark}
\newaliascnt{remark}{theorem}
\newtheorem{remark}[remark]{Remark}\aliascntresetthe{remark}
\crefname{lemma}{Lemma}{Lemmas}\Crefname{lemma}{Lemma}{Lemmas}
\crefname{proposition}{Proposition}{Propositions}\Crefname{proposition}{Proposition}{Propositions}
\crefname{corollary}{Corollary}{Corollaries}\Crefname{corollary}{Corollary}{Corollaries}
\crefname{remark}{Remark}{Remarks}\Crefname{remark}{Remark}{Remarks}
\newcommand{\C}{\mathbb C}\newcommand{\E}{\mathbb E}\newcommand{\R}{\mathbb R}\newcommand{\Z}{\mathbb Z}
\newcommand{\F}{\mathbb F}\newcommand{\Dens}{\mathcal D}
\newcommand{\Tr}{\operatorname{Tr}}\newcommand{\diag}{\operatorname{diag}}\newcommand{\rank}{\operatorname{rank}}
\newcommand{\vecop}{\operatorname{vec}}\newcommand{\id}{\mathrm{id}}
\newcommand{\HS}{\mathrm{HS}}
\newcommand{\ket}[1]{\lvert#1\rangle}\newcommand{\bra}[1]{\langle#1\rvert}\newcommand{\proj}[1]{\lvert#1\rangle\!\langle#1\rvert}

\newcommand{\Tch}{\mathsf T}
\newcommand{\etalchar}[1]{$^{#1}$}
\newcolumntype{Y}{>{\raggedright\arraybackslash}X}
\title{\textbf{Closed-Form Arbitrary-Gap Rényi Nonadditivity Above Any Positive Order via Direct Products of Free Groups and High-Girth Graphs}}
\author[1]{Guocheng Zhen}
\author[2]{Ranyiliu Chen\thanks{chenranyiliu@quantumsc.cn}}
\author[1]{Xin Wang\thanks{felixxinwang@hkust-gz.edu.cn}}
\affil[1]{\small Thrust of Artificial Intelligence, Information Hub,\par The Hong Kong University of Science and Technology (Guangzhou), Guangdong 511453, China}
\affil[2]{\small Quantum Science Center of Guangdong-Hong Kong-Macao Greater Bay Area, Shenzhen 518045, China}
\date{}

\newcommand{\Sh}{\mathfrak S}

\begin{document}
\maketitle
\begin{abstract}
We give a closed-form finite-dimensional construction of quantum channels with real algebraic entries exhibiting arbitrarily large violations of minimum-output R\'enyi entropy additivity uniformly above any prescribed positive R\'enyi order, using a different finite realization based on direct products of free groups and fixed-field high-girth graphs. More precisely, for every $p_0>0$ and $g>0$, we construct one channel and one Bell input for which the additivity gap is greater than $g$ simultaneously for all $p\in[p_0,\infty]$.  With $r=\lceil2g\rceil$, the construction uses $15r$ output qubits and
\(
O\!\left(r2^{30r}+r^2(1+p_0^{-1})\right)
\)
input qubits. Compared with the stated parameters of Shou and Gorshkov \cite{ShouGorshkov2026}, our construction uses fewer output qubits for every target gap and has a smaller asymptotic input-qubit bound as the gap grows with the positive cutoff fixed. At the von Neumann order, a separate specialization gives an explicit channel with gap greater than $1/1024$ using $37608898$ input qubits and $8$ output qubits.
\end{abstract}

\tableofcontents

\section{Introduction}\label{sec:intro}

For finite-dimensional quantum channels $\Phi$ and $\Psi$, product inputs always give
\[
  H_{\min,p}(\Phi\otimes\Psi)
  \le
  H_{\min,p}(\Phi)+H_{\min,p}(\Psi),
\]
and the minimum-output R\'enyi-entropy additivity problem asks whether equality must hold.  The von Neumann case $p=1$ is particularly important for classical communication through the Holevo--Schumacher--Westmoreland theorem and Shor's equivalence of additivity questions \cite{SchumacherWestmoreland1997,Holevo1998,Shor2004,ShorErratum2004}. Werner and Holevo first gave an explicit violation at sufficiently large R\'enyi orders \cite{WernerHolevo2002}; Hayden and Winter subsequently proved nonadditivity for every $p>1$ by random constructions \cite{HaydenWinter2008}.  Cubitt et al.\ treated $p=0$ and sufficiently small positive orders \cite{CubittEtAl2008}, while Hastings established the von Neumann case by a random finite-dimensional construction \cite{Hastings2009}.

Free-probability methods later provided a more structural approach to the problem.  Belinschi, Collins, and Nechita analyzed asymptotic output spectra of random channels and obtained Bell-state violations arbitrarily close to one bit at $p=1$ \cite{BelinschiCollinsNechita2012,BelinschiCollinsNechita2016}; related free-compression estimates were developed in \cite{CollinsFukudaZhong2015}.  Collins combined strong asymptotic freeness with Haagerup's inequality to obtain a concise proof for random mixed-unitary channels \cite{CollinsMale2014,Collins2018}, and Collins and Youn extended the Haagerup framework to products of free groups \cite{CollinsYoun2022}.  More recently, Leung, Lovitz, and Wu obtained finite-dimensional counterexamples for $0\le p<1/4$ and $p>3/4$ \cite{LeungLovitzWu2026}, while Zhen et al.\ used direct products of free groups to obtain a single finite-dimensional real channel and the same Bell input with an almost-one-bit violation simultaneously for every $p\in[0,\infty]$ \cite{ZhenEtAlAllOrders2026}.

Explicit and deterministic constructions have also developed rapidly. Derksen and Lovitz gave constructive counterexamples for every $p>1$ \cite{DerksenLovitz2026}.  At the von Neumann order, Lovitz and Wu and, independently, Zhen et al.\ used the deterministic near-free permutation framework of O'Donnell and Wu \cite{ODonnellWu2020} to obtain polynomial-time algorithmic constructions of nonadditive channels \cite{LovitzWu2026,ZhenEtAlDeterministic2026}.  More recently, Shou and Gorshkov obtained an explicit finite-dimensional construction, including arbitrarily large gaps uniformly on every prescribed $[p_0,\infty]$ with $p_0>0$ \cite{ShouGorshkov2026}, while Cheng and Wu constructed explicit channels with unbounded gains in classical communication using deterministic Clifford operations, measurement, and feedforward \cite{ChengWu2026}.

The present work builds on a common principle emerging from the recent constructions of Shou--Gorshkov and Cheng--Wu \cite{ShouGorshkov2026,ChengWu2026}: finite moment information can be combined with a suitable damping or measurement mechanism to control exceptional directions, thereby avoiding the need for strong convergence in the finite-dimensional realization.  We implement this principle in the direct-product free-group model using a different explicit realization.  Fixed-field high-girth graphs provide real permutation matrices that reproduce the required short free-group moments exactly, while filters built from Chebyshev polynomials and spherical averages convert these moment identities into the uniform one-copy bounds needed for the entropy comparison.  This yields a closed-form finite-dimensional construction with real algebraic entries and fully prescribed parameters.

The paper is mainly organized as follows. Section~\ref{sec:main} states the main results.  Section~\ref{sec:prelim} collects the required preliminaries, and Section~\ref{sec:free} isolates the abstract finite-dimensional estimates and entropy criterion.  Sections \ref{sec:finite} and \ref{sec:filter} construct, respectively, the high-girth finite realization and the filter, and verify the hypotheses of that criterion with explicit parameters. Section~\ref{sec:proof} completes the proofs of the main results. Section~\ref{sec:discussion} compares our finite-dimensional realization and dimension bounds with those of Shou and Gorshkov \cite{ShouGorshkov2026}.

\section{Main results}\label{sec:main}
All logarithms are natural unless a base is shown, so entropy gaps are measured in nats. For a channel $\Phi:M_n(\C)\to M_D(\C)$, put
\begin{equation}\label{eq:gap}
 \Delta_p(\Phi;\Omega_n)
 =2H_{\min,p}(\Phi)-H_p\bigl(\Phi^{\otimes2}(\Omega_n)\bigr),
 \qquad
 \Omega_n=\proj{\Omega_n},\quad
 \ket{\Omega_n}=n^{-1/2}\sum_{j=1}^n\ket{j,j}.
\end{equation}
A positive value implies minimum-output-entropy nonadditivity. The channels have the form
\begin{equation}\label{eq:channel}
 [\Phi_F(X)]_{st}
 =\frac1D\Tr(W_sFXFW_t^*)
  +\frac{\delta_{st}}D\Tr\bigl((I_n-F^2)X\bigr),
 \qquad 1\le s,t\le D,
\end{equation}
where the $W_s$ are prescribed real permutation matrices and $F$ is a prescribed real symmetric positive contraction, that is,
\[
  F=F^T,\qquad 0\preceq F\preceq I_n.
\]
The first summand is the complementary mixed-unitary channel applied to $FXF$, and the second is a completely positive replacement map. Their traces sum to $\Tr X$, so $\Phi_F$ is a channel.

\begin{theorem}[Uniform positive-order gaps with explicit finite dimensions]\label{thm:main}

For every finite $p_0>0$ and every integer $r\ge1$, the formulas \eqref{eq:finite-branches}, \eqref{eq:parameter-list}, and \eqref{eq:filter} prescribe a channel of the form \eqref{eq:channel},
\[
 \Phi_{r,p_0}:M_{2^{N(r,p_0)}}(\C)\longrightarrow M_{2^{15r}}(\C),
\]
with real algebraic entries, such that
\begin{equation}\label{eq:main-gap}
 \inf_{p\in[p_0,\infty]}
 \Delta_p(\Phi_{r,p_0};\Omega_{2^{N(r,p_0)}})>\frac r2.
\end{equation}
The integer $N(r,p_0)$ is given exactly in \eqref{eq:input-budget} and satisfies
\begin{equation}\label{eq:main-dimension-bound}
 N(r,p_0)
 \le 9216r\,2^{30r}
 +36864r\left(45r+20+\frac{r+16}{\min\{p_0,1/2\}}\right).
\end{equation}
In particular, for every $g>0$, the choice $r=\lceil2g\rceil$ gives one channel and the same Bell input with a gap greater than $g$ at every order in $[p_0,\infty]$.
\end{theorem}

Both the channel and its Bell input may depend on the prescribed
$p_0$ and $g$, but not on the running order $p$.  The construction in \Cref{thm:main} uses
\[
  O\bigl(r2^{30r}+r^2(1+p_0^{-1})\bigr)
\]
input qubits and $15r$ output qubits. For a fixed small positive gap
at the single order $p=1$, the separate specialization below uses
$37608898$ input qubits and eight output qubits.

\begin{samepage}
\begin{corollary}[An eight-output-qubit von Neumann example]\label{cor:von-neumann}
The matrices in \eqref{eq:small-branches} and \eqref{eq:small-filter} prescribe a real algebraic channel
\[
 \Phi_*:M_{2^{N_*}}(\C)\longrightarrow M_{256}(\C),
 \qquad N_*=37608898,
\]
whose standard Bell input satisfies
\begin{equation}\label{eq:small-gap}
 2H_{\min,1}(\Phi_*)-H_1\bigl(\Phi_*^{\otimes2}(\Omega_{2^{N_*}})\bigr)
 >\frac1{1024}.
\end{equation}
\end{corollary}
\end{samepage}

Applying \Cref{thm:main} directly to the target gap $g=1/1024$
would require
\(
r=\left\lceil\frac1{512}\right\rceil=1,
\)
and would therefore guarantee the much larger gap $1/2$.  The
corollary instead uses estimates tailored to the single order $p=1$
to obtain a smaller explicit instance. The numerical parameters in \Cref{thm:main,cor:von-neumann} are not intended to be optimal.  They are chosen to give transparent, fully analytic margins throughout the proof.  Accordingly, the resulting dimensions should be viewed as explicit certified bounds.

\section{Preliminaries}\label{sec:prelim}

In this section, we collect the notation and preliminary results needed for the constructions and proofs below.

\subsection{Channels, norms, and Bell vectors}
All Hilbert spaces and matrix algebras are complex. Channel input and output spaces are finite-dimensional; infinite-dimensional regular representations are used only for free norm estimates. We fix orthonormal bases throughout, with inner products conjugate-linear in the first argument. Write
\[
 \Dens(\C^n)=\{\rho\in M_n(\C):\rho\succeq0,\ \Tr\rho=1\},
 \qquad \tau_n=n^{-1}\Tr.
\]
The operator norm is denoted by $\|\cdot\|$, the Hilbert--Schmidt norm by $\|X\|_{\HS}=(\Tr X^*X)^{1/2}$, and the trace norm by $\|X\|_1=\Tr|X|$. A quantum channel is a completely positive trace-preserving linear map. For $0<p<\infty$, $p\ne1$, put
\[
 H_p(\rho)=\frac1{1-p}\log\Tr\rho^p,
 \qquad
 H_1(\rho)=-\Tr(\rho\log\rho),\qquad
 H_\infty(\rho)=-\log\|\rho\|,
\]
with $0\log0=0$. The minimum output entropy is
\[
 H_{\min,p}(\Phi)=\min_{\rho\in\Dens(\C^n)}H_p(\Phi(\rho)).
\]
The minimum exists by compactness and continuity, including $p=\infty$. Entropy decreases with the order and is additive on tensor products of states:
\begin{equation}\label{eq:entropy-elementary}
 H_p(\rho)\ge H_2(\rho)\quad(0<p\le2),
 \qquad
 H_p(\rho\otimes\sigma)=H_p(\rho)+H_p(\sigma).
\end{equation}
For $p\ge2$,
\begin{equation}\label{eq:purity-interpolation}
 \Tr\rho^p\le\|\rho\|^{p-2}\Tr\rho^2.
\end{equation}
These statements follow from the eigenvalue formulas and monotonicity of power means.

Our column-vectorization convention is
\[
 \vecop(A)=\sum_{j=1}^n Ae_j\otimes e_j.
\]
It gives
\begin{equation}\label{eq:vectorization}
 (A\otimes B)\ket{\Omega_n}=n^{-1/2}\vecop(AB^T),
 \qquad
 n^{-1}\langle\vecop(A),\vecop(B)\rangle=\tau_n(A^*B).
\end{equation}
For real orthogonal matrices, $B^T=B^{-1}$. A standard Bell input therefore associates an ordered pair of branches $(s,t)$ with the relative product $W_sW_t^{-1}$.

\subsection{Majorization and spectral perturbations}\label{sec:majorization}
For probability vectors $x,y$ of the same length, write $x\prec y$ when
\[
 \sum_{i=1}^j x_i^\downarrow\le\sum_{i=1}^j y_i^\downarrow
 \quad(1\le j<d),\qquad
 \sum_i x_i=\sum_i y_i.
\]
Tensor products preserve this order: $x\prec y$ implies $x^{\otimes r}\prec y^{\otimes r}$. R\'enyi entropies are Schur-concave at every order, so $x\prec y$ implies $H_p(x)\ge H_p(y)$. Schur--Horn gives
\begin{equation}\label{eq:schur-horn}
 \diag(A)\prec\lambda(A)
\end{equation}
for a Hermitian matrix $A$, with $\lambda(A)$ its eigenvalue vector. We use these standard facts and the Ky Fan variational principle in the form recorded in \cite[Chapters~II--III]{Bhatia1997}.

For states $\rho,\sigma\in M_d(\C)$, trace distance has the variational expression
\begin{equation}\label{eq:trace-distance}
 \frac12\|\rho-\sigma\|_1
 =\max_{0\preceq P\preceq I}\bigl|\Tr P(\rho-\sigma)\bigr|.
\end{equation}
In particular, the Ky Fan variational principle and \eqref{eq:trace-distance} imply
\begin{equation}\label{eq:ky-fan-trace-perturbation}
 \left|
 \sum_{i=1}^k\lambda_i^\downarrow(\rho)
 -
 \sum_{i=1}^k\lambda_i^\downarrow(\sigma)
 \right|
 \le
 \frac12\|\rho-\sigma\|_1,
 \qquad 1\le k\le d.
\end{equation}
Let $P_1,\ldots,P_m$ be mutually orthogonal projections with
\[
\sum_{j=1}^m P_j=I. 
\]
The associated \emph{pinching map} is
\begin{equation}\label{eq:pinching}
  \mathcal P(X)
  :=
  \sum_{j=1}^m P_jXP_j .
\end{equation}
If $\omega=e^{2\pi i/m}$ and
\[
  U=\sum_{j=1}^m \omega^j P_j,
\]
then
\begin{equation}\label{eq:pinching-random-unitary}
  \mathcal P(X)
  =
  \frac1m\sum_{r=0}^{m-1}U^rXU^{-r}.
\end{equation}
Thus pinching is a convex average of unitary conjugations.  More
generally, for unitaries $U_\alpha$ and probabilities $t_\alpha$,
\begin{equation}\label{eq:unitary-average-majorization}
  \lambda\!\left(
    \sum_\alpha t_\alpha U_\alpha\rho U_\alpha^*
  \right)
  \prec
  \lambda(\rho).
\end{equation}
In particular,
\begin{equation}\label{eq:pinching-majorization}
  \lambda(\mathcal P(\rho))
  \prec
  \lambda(\rho).
\end{equation}
Since R\'enyi entropy is Schur-concave,
\begin{equation}\label{eq:majorization-entropy}
  H_p(\mathcal P(\rho))
  \ge H_p(\rho),
  \qquad
  H_p\!\left(
    \sum_\alpha t_\alpha U_\alpha\rho U_\alpha^*
  \right)
  \ge H_p(\rho)
\end{equation}
for every R\'enyi order considered below.

\subsection{Free groups, tensor legs, and Haagerup's inequality}\label{sec:haagerup}
Let $\mathbb F_k=\langle g_1,\ldots,g_k\rangle$ be the free group. Its letters are $g_1^{\pm1},\ldots,g_k^{\pm1}$. A word is \emph{reduced} if adjacent letters are never inverse. Cancelling such pairs gives a unique reduced expression for each group element; the empty word is the identity $e$, and $|w|$ denotes reduced length. For example, $g_1g_2g_2^{-1}$ reduces to $g_1$, whereas $g_1^2$ is already a nonempty reduced word. 

For a discrete group $G$, the Hilbert space $\ell^2(G)$ has an orthonormal basis $(e_h)_{h\in G}$ indexed by group elements. The left regular representation and its canonical trace are
\[
 \lambda_G(g)e_h=e_{gh},\qquad
 \tau_G(T)=\langle e_e,Te_e\rangle,
 \qquad
 \tau_G(\lambda_G(g))=\mathbf1_{\{g=e\}}.
\]
Indeed, left multiplication by $g$ sends $e_e$ to $e_g$, so its diagonal coefficient at $e_e$ is zero unless $g=e$. The trace is faithful on the reduced group algebra.

For $G=\mathbb F_k^b$, an element is a tuple $(w_1,\ldots,w_b)$ and multiplication is componentwise. Thus elements supported in different components commute. Under the canonical identification $\ell^2(G)\cong\bigotimes_{j=1}^b\ell^2(\mathbb F_k)$,
\[
 \lambda_G(w_1,\ldots,w_b)=\bigotimes_{j=1}^b\lambda_{\mathbb F_k}(w_j),
 \qquad
 \tau_G\!\left(\bigotimes_{j=1}^b T_j\right)
 =\prod_{j=1}^b\tau_{\mathbb F_k}(T_j).
\]
Each tensor factor is called a \emph{tensor leg}. Letters on a single leg retain free-group reduction; operators on distinct legs commute. A finitely supported function $f$ is identified with
\[
 \lambda_G(f)=\sum_{w\in G}f(w)\lambda_G(w),\qquad
 \|f\|_2^2=\sum_w|f(w)|^2.
\]

\begin{lemma}[Product Haagerup inequality]\label{lem:haagerup}
Suppose $f$ is supported on elements $(w_1,\ldots,w_b)\in\mathbb F_k^b$ with fixed component lengths $|w_j|=a_j$. Then
\begin{equation}\label{eq:haagerup}
 \|\lambda_G(f)\|\le\prod_{j=1}^b(a_j+1)\,\|f\|_2.
\end{equation}
\end{lemma}
This is \cite[Proposition~3.2]{CollinsYoun2022}; for a single factor it is \cite[Lemma~1.4]{Haagerup1979}. The estimate is independent of the free rank $k$. Later, each coordinate length is $0$, $1$, or $2$. 

\subsection{Graphs, edge colours, and the coordinate family}
\label{sec:graph-background}

This subsection reviews the graph-theoretic ingredients needed for the
finite-dimensional realization developed later in the paper.

A finite simple graph is a pair $\Gamma=(V,E)$, where $V$ is its vertex
set and $E$ is a collection of two-element subsets of $V$.  We write
$v\sim w$ when $\{v,w\}\in E$.  The graph is \emph{bipartite} if
\[
  V=\mathcal P\sqcup\mathcal L
\]
and every edge joins a vertex of $\mathcal P$ to one of $\mathcal L$. A walk $(v_0,v_1,\ldots,v_m)$ is \emph{nonbacktracking} if $v_{j+1}\neq v_{j-1}$ whenever both sides are defined.  The \emph{girth} is the length of a shortest cycle.

We use the coordinate graphs $D(\nu,Q)$ of
\cite[Section~2]{LazebnikUstimenkoWoldar1995}, where $Q$ is a prime power
and $\nu\ge2$.  Their two vertex classes are distinct copies of
$\F_Q^\nu$:
\begin{equation}\label{eq:coordinate-vertex-classes}
  \mathcal P=\{(x):x\in\F_Q^\nu\},
  \qquad
  \mathcal L=\{[x]:x\in\F_Q^\nu\}.
\end{equation}
Each vertex class is a copy of $\F_Q^\nu$, but, following
\cite[Section~2]{LazebnikUstimenkoWoldar1995}, we label its coordinates
according to the recursive structure of the incidence relations rather
than by $u_1,\ldots,u_\nu$.  The coordinates are taken from the ordered
list
\begin{equation}\label{eq:graph-coordinates}
  u_1,\;
  u_{11},u_{12},u_{21},\;
  u_{22},u'_{22},u_{23},u_{32},\;
  u_{33},u'_{33},u_{34},u_{43},\ldots,
\end{equation}
where, from level $i=2$ onward, the four coordinates
\[
  u_{ii},\quad u'_{ii},\quad u_{i,i+1},\quad u_{i+1,i}
\]
form the next block.  The graph $D(\nu,Q)$ uses the first $\nu$
coordinates in this ordered list.  For example, a vertex of $D(4,Q)$
has coordinates
\[
  (u_1,u_{11},u_{12},u_{21}),
\]
whereas a vertex of $D(8,Q)$ has coordinates
\[
  (u_1,u_{11},u_{12},u_{21},
    u_{22},u'_{22},u_{23},u_{32}).
\]
For a point these coordinates are denoted by
$p_1,p_{11},\ldots$, and for a line by $l_1,l_{11},\ldots$. For a point $p$ and a line $l$, $p\sim l$ if
\begin{equation}\label{eq:incidence}
\begin{alignedat}{3}
 l_{11}-p_{11}&=l_1p_1,\\
 l_{12}-p_{12}&=l_{11}p_1,\\
 l_{21}-p_{21}&=l_1p_{11},\\[1mm]
 l_{ii}-p_{ii}&=l_1p_{i-1,i},
 &\qquad&\\
 l'_{ii}-p'_{ii}&=l_{i,i-1}p_1,
 &&\qquad i\ge2,\\
 l_{i,i+1}-p_{i,i+1}&=l_{ii}p_1,
 &&\\
 l_{i+1,i}-p_{i+1,i}&=l_1p'_{ii}.
 &&
\end{alignedat}
\end{equation}
Only equations whose left-hand coordinate occurs among the first $\nu$
coordinates in \eqref{eq:graph-coordinates} are imposed.  For fixed
$p$ and $l_1\in\F_Q$, these equations determine $l$ successively; the
first steps are
\begin{equation}\label{eq:incidence-initial}
 \begin{aligned}
  l_{11}&=p_{11}+l_1p_1,\\
  l_{12}&=p_{12}+l_{11}p_1,\\
  l_{21}&=p_{21}+l_1p_{11}.
 \end{aligned}
\end{equation}
Conversely, fixing $l$ and $p_1$ determines $p$.  Hence every vertex has
exactly $Q$ neighbours.

For an edge $\{p,l\}$ define its colour by
\begin{equation}\label{eq:edge-colour}
  c=p_1+l_1\in\F_Q.
\end{equation}
Given a vertex and a colour $c$, the first coordinate of its unique
neighbour is determined by $l_1=c-p_1$ or $p_1=c-l_1$, and
\eqref{eq:incidence} then determines all remaining coordinates.
Hence every vertex is incident to exactly one edge of each of the
$Q$ colours.  In particular, every vertex has degree $Q$, so
$D(\nu,Q)$ is $Q$-regular.

Fix now a colour $c\in\F_Q$.  Since every point and every line is
incident to exactly one colour-$c$ edge, the colour-$c$ edges are
pairwise disjoint and cover all vertices.  They therefore form a
perfect matching between the point and line classes.  Let
\begin{equation}\label{eq:matching-swap}
  J_c(v)
  =
  \text{the other endpoint of the unique colour-$c$ edge incident to }v .
\end{equation}
Then $J_c$ is an involutive permutation of the vertex set:
$J_c^2=\id$.

For illustration, $D(2,2)$ has four points and four lines, and adjacency
reduces to
\[
  l_{11}-p_{11}=l_1p_1
  \qquad\text{in }\F_2.
\]
For example, from $p=(1,0)$ the choices $l_1=0,1$ give the two
neighbours $[0,0]$ and $[1,1]$, of colours $1$ and $0$, respectively.
The resulting graph is the eight-cycle shown in
Figure~\ref{fig:colours}.

\begin{figure}[htbp]
\centering
\begin{tikzpicture}[scale=0.98,>=Stealth,
 pt/.style={circle,draw,fill=white,minimum size=9mm,inner sep=1pt},
 ln/.style={rectangle,draw,fill=black!7,minimum size=9mm,inner sep=1pt},
 rededge/.style={->,thick,draw=red!65!black},
 blueedge/.style={->,thick,dashed,draw=blue!65!black}]
\node[pt] (p0) at (90:2.15) {\small$(0,0)$};
\node[ln] (l0) at (45:2.15) {\small$[0,0]$};
\node[pt] (p1) at (0:2.15) {\small$(1,0)$};
\node[ln] (l1) at (-45:2.15) {\small$[1,1]$};
\node[pt] (p2) at (-90:2.15) {\small$(0,1)$};
\node[ln] (l2) at (-135:2.15) {\small$[0,1]$};
\node[pt] (p3) at (180:2.15) {\small$(1,1)$};
\node[ln] (l3) at (135:2.15) {\small$[1,0]$};
\foreach \i in {0,1,2,3}{\draw[rededge] (p\i)--(l\i);}
\draw[blueedge] (l0)--(p1);
\draw[blueedge] (l1)--(p2);
\draw[blueedge] (l2)--(p3);
\draw[blueedge] (l3)--(p0);
\node[align=center,font=\small] at (0,0)
 {$D(2,2)$\\$l_{11}-p_{11}=l_1p_1$};
\draw[rededge] (3.0,0.9)--(4.0,0.9)
 node[right,black,font=\small] {colour $0$};
\draw[blueedge] (3.0,0.2)--(4.0,0.2)
 node[right,black,font=\small] {colour $1$};
\end{tikzpicture}
\caption{The coordinate graph $D(2,2)$.  Circles and squares represent points and lines, while the two edge styles represent the two colours.
}
\label{fig:colours}
\end{figure}
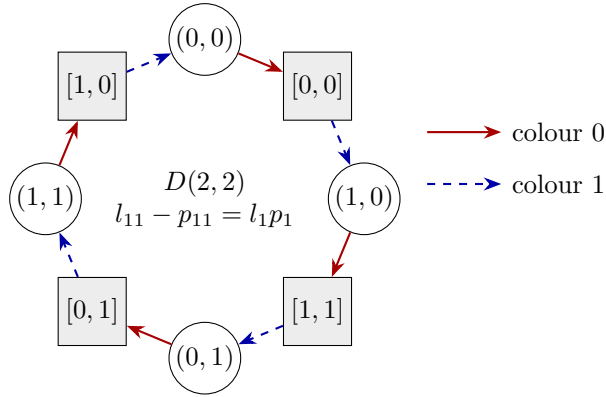

An individual $J_c$ cannot model a free generator, since
$J_c^2=\id$.  When $Q$ is even, fix an ordering
$c_0,\ldots,c_{Q-1}$ of $\F_Q$ and pair the colours as
\[
  (\alpha_i,\beta_i)
  =(c_{2i-2},c_{2i-1}),
  \qquad 1\le i\le Q/2.
\]
Define
\begin{equation}\label{eq:general-paired-rule}
  u_i(v)=
  \begin{cases}
   J_{\alpha_i}(v),&v\in\mathcal P,\\
   J_{\beta_i}(v),&v\in\mathcal L.
  \end{cases}
\end{equation}
Then $u_i$ is a permutation, with inverse
\begin{equation}\label{eq:general-paired-inverse}
  u_i^{-1}(v)=
  \begin{cases}
   J_{\beta_i}(v),&v\in\mathcal P,\\
   J_{\alpha_i}(v),&v\in\mathcal L.
  \end{cases}
\end{equation}
Let $q=|V|$, and let $(e_v)_{v\in V}$ be the standard basis of
$\C^q$.  Each permutation $u_i$ determines a permutation matrix
$P_i\in O(q)$ by
\begin{equation}\label{eq:general-permutation-matrix}
  P_i e_v=e_{u_i(v)},
  \qquad 1\le i\le Q/2.
\end{equation}

Thus, once the coloured graph $D(\nu,Q)$ and the pairing of its
$Q$ colours are fixed, we obtain $Q/2$ distinguished permutation
matrices $P_1,\ldots,P_{Q/2}$.  Let
\[
  \mathbb F_{Q/2}
  =
  \langle g_1,\ldots,g_{Q/2}\rangle
\]
be the free group on $Q/2$ generators.  The assignment
\[
  g_i\longmapsto P_i
\]
extends uniquely to a representation
\begin{equation}\label{eq:graph-representation}
  \pi:\mathbb F_{Q/2}\longrightarrow O(q).
\end{equation}
Accordingly, for
\[
  w=g_{i_1}^{\varepsilon_1}\cdots
    g_{i_s}^{\varepsilon_s},
  \qquad \varepsilon_j\in\{\pm1\},
\]
we write
\begin{equation}\label{eq:word-walk-convention}
  P_w:=\pi(w)
  =
  P_{i_1}^{\varepsilon_1}\cdots
  P_{i_s}^{\varepsilon_s}.
\end{equation}

Each letter $g_i^{\pm1}$ therefore acts through the permutation
$u_i^{\pm1}$ and traverses exactly one graph edge.  By construction,
two consecutive steps traverse the same edge in opposite directions
if and only if the corresponding letters are inverse.  Hence a
reduced word determines a nonbacktracking walk of the same length. The girth condition turns this correspondence into exact short-word trace matching.

\begin{lemma}[Short-word trace matching]
\label{lem:coloured-short-traces}
Assume that $Q$ is even, and let
\[
  \pi:\mathbb F_{Q/2}\longrightarrow O(2Q^\nu),
  \qquad
  \pi(g_i)=P_i,
\]
be the representation associated with the paired-colour permutations
of $D(\nu,Q)$.  If
\[
  \operatorname{girth} D(\nu,Q)>L,
\]
then, for every $w\in\mathbb F_{Q/2}$ with $|w|\le L$,
\begin{equation}\label{eq:general-short-traces}
  \tau_{2Q^\nu}(P_w)
  =
  \tau_{\mathbb F_{Q/2}}(\lambda(w))
  =
  \mathbf 1_{\{w=e\}}.
\end{equation}
\end{lemma}

\begin{proof}
It suffices to consider $w\neq e$.  Write a reduced expression
\[
  w=x_s\cdots x_1,
  \qquad
  x_j=g_{i_j}^{\varepsilon_j},
  \qquad
  s=|w|\le L.
\]
For $v_0\in V(D(\nu,Q))$, define
\[
  v_j=u_{i_j}^{\varepsilon_j}(v_{j-1}),
  \qquad 1\le j\le s.
\]
Then
\[
  P_we_{v_0}=e_{v_s}.
\]
Since $w$ is reduced, the walk
\[
  v_0,v_1,\ldots,v_s
\]
is nonbacktracking.  Since $P_w$ is a permutation matrix,
\[
  \Tr(P_w)
  =
  \#\{v\in V(D(\nu,Q)):P_we_v=e_v\}.
\]
Suppose that $P_we_{v_0}=e_{v_0}$. Then $v_s=v_0$, so the
nonbacktracking walk is closed and contains a cycle of length at most
$s\le L<\operatorname{girth}D(\nu,Q)$, a contradiction. Hence no
basis vector is fixed and $\tau_{2Q^\nu}(P_w)=0$. For $w=e$,
$P_w=I$, completing the proof.
\end{proof}

The full graph $D(\nu,Q)$ has $2Q^\nu$ vertices.  To pass to a smaller
vertex set without losing the coloured-neighbour structure, we use the
coordinate invariants from the proof of
\cite[Lemma~3.1]{LazebnikUstimenkoWoldar1995}.  For notational
uniformity, introduce the auxiliary boundary conventions
\begin{equation}\label{eq:aux-coordinates}
\begin{aligned}
 p_{00}=l_{00}&=-1,&
 p'_{00}=l'_{00}&=1,\\
 p'_{11}&=p_{11},&
 l'_{11}&=l_{11},\\
 p_{01}&=p_1,&
 l_{10}&=l_1,\\
 p_{10}=l_{01}&=0,&
 p_{0,-1}=l_{0,-1}&=0.
\end{aligned}
\end{equation}
These are conventions, not additional free coordinates.

Set
\[
  b=\left\lfloor\frac{\nu+2}{4}\right\rfloor
\]
and, for $2\le j\le b$, define
\begin{equation}\label{eq:graph-invariants}
  a_j(u)
  =
  \sum_{i=0}^{j}
  \bigl(
    u_{ii}u'_{j-i,j-i}
    -
    u_{i,i+1}u_{j-i,j-i-1}
  \bigr).
\end{equation}
The second product at $i=j$ is defined to be zero, consistently with
$u_{0,-1}=0$; thus no unretained coordinate $u_{j,j+1}$ is needed.
The relevant property is that these quantities are invariant across
edges:
\begin{equation}\label{eq:adjacency-invariance}
  p\sim l
  \quad\Longrightarrow\quad
  a_j(p)=a_j(l),
  \qquad
  2\le j\le b.
\end{equation}
For
\[
  \xi=(\xi_2,\ldots,\xi_b)\in\F_Q^{b-1},
\]
define
\begin{equation}\label{eq:general-invariant-fibre}
  V_\xi
  =
  \{u:a_j(u)=\xi_j,\ 2\le j\le b\}.
\end{equation}
If $u\in V_\xi$ and $v\sim u$, then
\eqref{eq:adjacency-invariance} gives
\[
  a_j(v)=a_j(u)=\xi_j,
  \qquad 2\le j\le b,
\]
so $v\in V_\xi$.  Thus every neighbour of a vertex in $V_\xi$
also lies in $V_\xi$.  The induced subgraph on $V_\xi$ therefore
retains all $Q$ coloured edges at every vertex.  In particular,
each $J_c$ preserves $V_\xi$, and hence so does every paired-colour
permutation $u_i$ when $Q$ is even.
Moreover, passing to an induced subgraph cannot create a shorter
cycle, so
\[
  \operatorname{girth}(D(\nu,Q)|_{V_\xi})
  \ge
  \operatorname{girth}D(\nu,Q).
\]

The fibres admit an explicit parametrization.  Isolating
$u'_{jj}$ in the equation $a_j(u)=\xi_j$ gives
\begin{equation}\label{eq:general-fibre-recursion}
  u'_{jj}
  =
  u_{jj}
  +\sum_{i=1}^{j-1}
     u_{ii}u'_{j-i,j-i}
  -\sum_{i=0}^{j-1}
     u_{i,i+1}u_{j-i,j-i-1}
  -\xi_j.
\end{equation}
Thus, after all coordinates except
$u'_{22},\ldots,u'_{bb}$ have been chosen freely,
\eqref{eq:general-fibre-recursion} determines these $b-1$ coordinates
successively.  Hence each vertex class has $Q^{\nu-b+1}$ elements and
\begin{equation}\label{eq:general-fibre-count}
  V_\xi
  \cong
  \{\mathrm P,\mathrm L\}\times\F_Q^{\nu-b+1},
  \qquad
  |V_\xi|=2Q^{\nu-b+1}.
\end{equation}

We collect the classical properties of $D(\nu,Q)$ used in the sequel. The girth estimate \eqref{eq:LUW-girth} is
\cite[Proposition~2.1(iii)]{LazebnikUstimenkoWoldar1995}, while
\eqref{eq:LUW-invariance} is established in the proof of
\cite[Lemma~3.1]{LazebnikUstimenkoWoldar1995}. 

\begin{proposition}[Classical properties of the coordinate graph]
\label{prop:LUW}
Let $Q$ be a prime power and $\nu\ge2$.  Then $D(\nu,Q)$ is simple,
bipartite, and $Q$-regular, with $2Q^\nu$ vertices.  If $\nu\ge3$ is
odd, then
\begin{equation}\label{eq:LUW-girth}
  \operatorname{girth}D(\nu,Q)\ge\nu+5.
\end{equation}
Moreover,
\begin{equation}\label{eq:LUW-invariance}
  p\sim l
  \quad\Longrightarrow\quad
  a_j(p)=a_j(l),
  \qquad
  2\le j\le
  \left\lfloor\frac{\nu+2}{4}\right\rfloor .
\end{equation}
\end{proposition}

\subsection{A fixed presentation of the field}\label{sec:fields}
The preceding graph construction was defined over any finite field. For the prescribed matrices we need both a fixed field and a fixed ordering of its elements, so that vertices and colour pairs can be indexed without choices. We use a field $\mathbb K$ of order $256$, described explicitly below. Its elements only label permutations; the channel matrices still act on complex Hilbert spaces.

Starting from $\F_2$, successively adjoin $t_1,t_2,t_3$ satisfying
\begin{equation}\label{eq:field256}
 t_1^2+t_1+1=0,\qquad
 t_2^2+t_2+t_1=0,\qquad
 t_3^2+t_3+t_1t_2=0.
\end{equation}
This gives a field $\mathbb K$ of order $256$. Indeed, over a finite field of characteristic two, $x^2+x+c$ is irreducible exactly when $c$ has absolute trace one \cite{LidlNiederreiter1997}. The map $x\mapsto x^2+x$ has kernel $\{0,1\}$ and image the trace-zero subspace. For \eqref{eq:field256},
\[
 \operatorname{Tr}_{\F_4/\F_2}(t_1)=1,\qquad
 \operatorname{Tr}_{\F_{16}/\F_2}(t_1t_2)
 =\operatorname{Tr}_{\F_4/\F_2}(t_1)=1,
\]
because $\operatorname{Tr}_{\F_{16}/\F_4}(t_2)=1$. The successive degrees are therefore $2,2,2$.

Order the field by binary coefficients in the basis
\[
 1,t_1,t_2,t_1t_2,t_3,t_1t_3,t_2t_3,t_1t_2t_3.
\]
The coefficient vector $(b_0,\ldots,b_7)$ has index $\sum_{j=0}^7b_j2^j$. Denote the resulting ordered elements by $c_0,\ldots,c_{255}$. This fixes all field operations and orderings used in the finite model.

\subsection{Partitions, Gaussian moments, and spherical averages}
\label{sec:moments}

We record the spherical moment formulas used later in the construction of the filter $F$ in \eqref{eq:channel}. For $j\ge1$, let $\mathcal P_2(2j)$ denote the set of pair partitions
of $\{1,\ldots,2j\}$.  Thus every
$\pi\in\mathcal P_2(2j)$ is a partition into $j$ two-element blocks.
If $G_1,\ldots,G_d$ are independent standard real Gaussian variables,
Isserlis's formula \cite{Isserlis1918} gives
\begin{equation}\label{eq:isserlis}
  \E\!\left[\prod_{r=1}^{2j}G_{i_r}\right]
  =
  \sum_{\pi\in\mathcal P_2(2j)}
  \prod_{\{a,b\}\in\pi}\delta_{i_a i_b},
\end{equation}
while all odd moments vanish.

Let $\sigma_d$ be the normalized surface measure on
\[
  S_{\R}^{d-1}
  =
  \{x\in\R^d:\|x\|_2=1\}.
\]
Writing a standard Gaussian vector as $G=Rx$, with
$x\sim\sigma_d$, the radius $R$ is independent of $x$ and satisfies
\[
  \E R^{2j}
  =
  d(d+2)\cdots(d+2j-2).
\]
Hence
\begin{equation}\label{eq:real-spherical-moment}
  \int
  x_{i_1}\cdots x_{i_{2j}}\,
  d\sigma_d(x)
  =
  \frac{
    \displaystyle
    \sum_{\pi\in\mathcal P_2(2j)}
    \prod_{\{a,b\}\in\pi}\delta_{i_a i_b}
  }{
    d(d+2)\cdots(d+2j-2)
  }.
\end{equation}
Consequently, for fixed Hermitian matrices
$X_1,\ldots,X_d$ and
\[
  T_x=\sum_{a=1}^d x_aX_a,
\]
one has
\begin{equation}\label{eq:wick-real}
  \int T_x^{2j}\,d\sigma_d(x)
  =
  \frac{
    \displaystyle
    \sum_{i_1,\ldots,i_{2j}=1}^d
    \left(
      \sum_{\pi\in\mathcal P_2(2j)}
      \prod_{\{a,b\}\in\pi}\delta_{i_a i_b}
    \right)
    X_{i_1}\cdots X_{i_{2j}}
  }{
    d(d+2)\cdots(d+2j-2)
  }.
\end{equation}
The order of the matrix factors in
\eqref{eq:wick-real} is unchanged.

We also use the corresponding complex spherical moments.  Let
$\nu_D$ be the normalized measure on
\[
  S_{\C}^{D-1}
  =
  \{z\in\C^D:\|z\|_2=1\}.
\]
For $j\ge1$,
\begin{equation}\label{eq:wick-complex}
  \int
  z_{i_1}\cdots z_{i_j}
  \overline z_{k_1}\cdots\overline z_{k_j}
  \,d\nu_D(z)
  =
  \frac{
    \displaystyle
    \sum_{\pi\in\Sh_j}
    \prod_{a=1}^j
    \delta_{i_a k_{\pi(a)}}
  }{
    D(D+1)\cdots(D+j-1)
  }.
\end{equation}
Here $\Sh_j$ is the symmetric group on $\{1,\ldots,j\}$.
This is the complex analogue of
\eqref{eq:real-spherical-moment}; it follows by applying the Gaussian
moment formula to standard circular complex Gaussian variables and
separating radius from direction.  In particular, after polynomial
expansion, \eqref{eq:wick-real} and \eqref{eq:wick-complex} express
the spherical averages used below as finite sums with explicit
coefficients.

Finally, we record the two spherical-cap formulas needed in the filter
estimates.  The first coordinate of a uniformly distributed point on
$S_{\R}^{d-1}$ has density
\begin{equation}\label{eq:real-density}
  c_d(1-t^2)^{(d-3)/2},
  \qquad -1<t<1,
  \qquad
  c_d
  =
  \frac{\Gamma(d/2)}
       {\sqrt{\pi}\,\Gamma((d-1)/2)}
  \ge\frac12
  \quad(d\ge3).
\end{equation}
For $D\ge2$ and any unit vector $v\in\C^D$,
\begin{equation}\label{eq:complex-cap}
  \nu_D
  \left\{
    z:
    |\langle v,z\rangle|^2\ge c
  \right\}
  =
  (1-c)^{D-1},
  \qquad
  0\le c\le1.
\end{equation}

\subsection{Chebyshev polynomials and algebraic functional calculus}
\label{sec:chebyshev-background}

We record the two elementary ingredients used below in the construction of the filter $F$. Let $\Tch_m$ denote the Chebyshev polynomial of the first kind, defined by
\[
  \Tch_0(x)=1,\qquad
  \Tch_1(x)=x,\qquad
  \Tch_{m+1}(x)=2x\Tch_m(x)-\Tch_{m-1}(x).
\]
For $m\ge1$,
\begin{equation}\label{eq:chebyshev-bound}
  |\Tch_m(x)|\le1,
  \qquad -1\le x\le1,
\end{equation}
and, for $u>1$,
\begin{equation}\label{eq:chebyshev-growth}
  \Tch_m\!\left(\frac{u+u^{-1}}2\right)
  =
  \frac{u^m+u^{-m}}2.
\end{equation}
All roots of $\Tch_m$ are real and lie in $(-1,1)$.  Consequently,
if $T=T^*$ and $B>0$, then
\begin{equation}\label{eq:chebyshev-matrix-bound}
  \|T\|\le B
  \quad\Longrightarrow\quad
  0\preceq
  \Tch_m(T/B)^2
  \preceq I.
\end{equation}
The recurrence also shows that $\Tch_m$ has integer coefficients, so
$\Tch_m(T/B)^2$ is a finite matrix polynomial whenever $B$ is
algebraic. We will also use the following elementary algebraic functional-calculus fact.  If $A$ is a finite-dimensional real symmetric positive definite matrix with real algebraic entries, then its unique positive inverse square root $A^{-1/2}$ is again real symmetric with real algebraic entries.

\section{A finite-dimensional entropy criterion}\label{sec:free}

We record the one-copy and Bell-output estimates needed for certification, independently of the subsequent graph and filter constructions. The free norm bounds and Bell quotient are taken from \cite{ZhenEtAlAllOrders2026,ShouGorshkov2026}.

\subsection{The free model and its norm estimates}
Let $k\ge1$, $\gamma\ge2$, and $r\ge1$ be integers. Put
\begin{equation}\label{eq:product-model}
 G_0=\mathbb F_k^\gamma,\qquad
 G=G_0^r=\prod_{\beta=1}^r\prod_{\alpha=1}^{\gamma}\mathbb F_k^{(\alpha,\beta)},
 \qquad K=2\gamma k,\quad D=K^r.
\end{equation}
In one block use the inverse-closed alphabet
\[
 \mathcal S=\{g_{j,\alpha}^{\varepsilon}:1\le j\le k,\ 1\le\alpha\le\gamma,
                         \ \varepsilon\in\{+1,-1\}\},
 \qquad \mathcal W=\mathcal S^r.
\]
For
\[
  s=(s_1,\ldots,s_r)\in\mathcal W,
  \qquad
  s_\beta=g_{j_\beta,\alpha_\beta}^{\varepsilon_\beta}\in\mathcal S,
\]
define $g_s\in G$ by
\begin{equation}\label{eq:branch-group-element}
  (g_s)_{\alpha,\beta}
  =
  \begin{cases}
    g_{j_\beta}^{\varepsilon_\beta},
      & \alpha=\alpha_\beta,\\
    e,
      & \alpha\neq\alpha_\beta,
  \end{cases}
  \qquad
  1\le\alpha\le\gamma,\quad
  1\le\beta\le r.
\end{equation}
Then
\begin{equation}\label{eq:free-branch-operator}
  u_s:=\lambda_G(g_s).
\end{equation} Order $\mathcal S$ lexicographically by $(\alpha,j,\varepsilon)$, with $+1$ before $-1$, and give $\mathcal W$ the induced order.

Let
\[
 \mathcal A_D=\{A\in M_D(\C):A=A^*,\ \diag A=0\},
 \qquad d=\dim_\R\mathcal A_D=D(D-1).
\]
The free tests are
\begin{equation}\label{eq:free-tests}
 T_A^{\mathrm f}=\sum_{s\ne t}A_{st}u_s^*u_t,
 \qquad L_z^{\mathrm f}=\sum_s z_su_s,
 \qquad A\in\mathcal A_D,\quad z\in\C^D.
\end{equation}
The product Haagerup estimate gives the following bounds, in the form proved in \cite[Lemma~B.2, Eqs.~(B.6), (B.14)--(B.15)]{ShouGorshkov2026}:
\begin{equation}\label{eq:free-bounds}
 \|T_A^{\mathrm f}\|\le J\|A\|_{\HS},\qquad
 \|L_z^{\mathrm f}\|\le(4\gamma)^{r/2}\|z\|_2,
 \qquad J^2=(K+16\gamma^2-7\gamma)^r-K^r.
\end{equation}

\subsection{One-copy estimates}
Let $\pi:G\to O(n)$ be a real orthogonal representation and $W_s=\pi(g_s)$. For a real symmetric matrix $0\preceq F\preceq I_n$, define $\Phi_F$ by \eqref{eq:channel} and set
\begin{equation}\label{eq:finite-tests}
 T_A=\sum_{s\ne t}A_{st}W_s^*W_t,
 \qquad L_z=\sum_s z_sW_s.
\end{equation}

The following proposition converts the two compression estimates into
the one-copy spectral and R\'enyi-entropy bounds used in the
nonadditivity certificates below.

\begin{proposition}[One-copy spectral control]
\label{prop:one-copy}
Let $\pi:G\to O(n)$ be a real orthogonal representation, let
$W_s=\pi(g_s)$, and let $F\in M_n(\R)$ satisfy
\[
  F=F^T,
  \qquad
  0\preceq F\preceq I_n.
\]
Suppose there exist constants $C,B\ge0$ such that, for every
$A\in\mathcal A_D$ and $z\in\C^D$,
\begin{equation}\label{eq:compression-interface}
  \|FT_AF\|
  \le
  C\|A\|_{\HS},
  \qquad
  FL_z^*L_zF
  \preceq
  B\|z\|_2^2 I_n.
\end{equation}
Set
\begin{equation}\label{eq:one-copy-constants}
  a_0:=1+\frac{C^2}{D},
  \qquad
  b_0:=1+B.
\end{equation}

Then every output state
\[
  \rho\in\Phi_F\bigl(\mathcal D(\C^n)\bigr)
\]
satisfies
\begin{equation}\label{eq:one-copy-spectral}
  \left\|\rho-\frac{I_D}{D}\right\|_{\HS}
  \le\frac{C}{D},
  \qquad
  \Tr\rho^2
  \le\frac{a_0}{D},
  \qquad
  \|\rho\|
  \le\frac{b_0}{D}.
\end{equation}
Consequently,
\begin{equation}\label{eq:one-copy-entropy}
  H_{\min,p}(\Phi_F)
  \ge
  \begin{cases}
    \log D-\log a_0,
      & 0<p\le2,\\[1mm]
    \displaystyle
    \log D-
    \frac{\log a_0+(p-2)\log b_0}{p-1},
      & 2<p<\infty,\\[2mm]
    \log D-\log b_0,
      & p=\infty.
  \end{cases}
\end{equation}
\end{proposition}

\begin{proof}
This is the filtered version of the quadratic-radius argument in \cite[Section~4.1]{ZhenEtAlAllOrders2026}, with the separate peak bound of \cite[Proposition~B.3]{ShouGorshkov2026}. For $\rho=\Phi_F(X)$,
\[
 \rho_{ss}=D^{-1},\qquad
 \Tr\bigl(A(\rho-I_D/D)\bigr)=D^{-1}\Tr(XFT_AF).
\]
Hilbert--Schmidt duality gives the first bound, and
$\Tr\rho^2=D^{-1}+\|\rho-I_D/D\|_{\HS}^2$ gives the second. For $\|z\|_2=1$,
\[
 z^*\rho z=D^{-1}\Tr(XFL_{\bar z}^*L_{\bar z}F)
            +D^{-1}\Tr((I_n-F^2)X)\le(B+1)/D.
\]
The entropy bounds follow from \eqref{eq:entropy-elementary} and \eqref{eq:purity-interpolation}.
\end{proof}

\begin{remark}
The first inequality in \eqref{eq:compression-interface} alone yields
the Hilbert--Schmidt and purity bounds in
\eqref{eq:one-copy-spectral}, and hence the entropy estimate for
$0<p\le2$.  The second inequality is needed only for the
largest-eigenvalue bound and the estimates at $p>2$.
\end{remark}

\subsection{The Bell profile and its stability}

We next identify the two-copy Bell output before filtering and then
control the perturbation caused by the filter.  The relevant
combinatorial object is the relative product of two branch labels. For $a,b\in\mathcal S\subset G_0$, define their one-block relative
product by
\begin{equation}\label{eq:relative-product-one-block}
  \operatorname{rel}_0(a,b):=ab^{-1}\in G_0.
\end{equation}
A \emph{relative-product fibre} is a level set of the map
$\operatorname{rel}_0:\mathcal S^2\to G_0$.  Thus two ordered branch
pairs $(a,b)$ and $(c,d)$ belong to the same fibre precisely when
\[
  ab^{-1}=cd^{-1}.
\]
Such coincidences will become identical Bell branches in the
unfiltered two-copy output.

The fibres can be read directly from the normal form in the direct
product of free groups.  The identity fibre consists of the $K$ pairs
$(a,a)$.  If $a$ and $b$ lie in different free factors, then they
commute and
\[
  ab^{-1}=b^{-1}a
         =\operatorname{rel}_0(b^{-1},a^{-1}),
\]
giving fibres of size two.  There are
\[
  \frac{\gamma-1}{2\gamma}K^2
\]
such fibres.  Finally, if $a\neq b$ lie in the same free factor,
uniqueness of reduced words gives a singleton fibre; their number is
\[
  \frac{K^2}{\gamma}-K.
\]
Consequently, the number of distinct one-block relative products is
\begin{equation}\label{eq:relative-product-count}
  M_{\gamma,K}
  =
  \frac{\gamma+1}{2\gamma}K^2-K+1.
\end{equation}
Equivalently, the associated one-block Bell profile, padded to length
$K^2$, is
\begin{equation}\label{eq:base-profile}
  v_{\gamma,K}
  =
  \left(
    \frac1K,\,
    \left(\frac2{K^2}\right)^{\!\times
      (\gamma-1)K^2/(2\gamma)},\,
    \left(\frac1{K^2}\right)^{\!\times(K^2/\gamma-K)},\,
    0^{\times(K^2-M_{\gamma,K})}
  \right).
\end{equation}
This relative-product quotient is the one computed in
\cite[Corollary~4.5]{ZhenEtAlAllOrders2026}; its product form is also
recorded in \cite[Lemma~B.5]{ShouGorshkov2026}.

For $s,t\in\mathcal W=\mathcal S^r$, the corresponding relative
product is
\begin{equation}\label{eq:relative-product}
  \operatorname{rel}(s,t)
  :=
  g_sg_t^{-1}
  =
  \bigl(
    \operatorname{rel}_0(s_\beta,t_\beta)
  \bigr)_{\beta=1}^r
  \in G_0^r.
\end{equation}
Since
\[
  \operatorname{rel}(s,t)
  =
  \bigl(
    \operatorname{rel}_0(s_\beta,t_\beta)
  \bigr)_{\beta=1}^r,
\]
two branch pairs $(s,t)$ and $(s',t')$ have the same relative product
if and only if this equality holds in every block.  Hence each global
relative-product fibre is the Cartesian product of the corresponding
one-block fibres, and its cardinality is the product of their
cardinalities.

We now explain the finite-dimensional condition needed to transfer
this combinatorial profile exactly.  Write
\[
  \Phi_0:=\Phi_{I_n},
  \qquad
  Z_0:=\Phi_0^{\otimes2}(\Omega_n),
  \qquad
  Z_F:=\Phi_F^{\otimes2}(\Omega_n).
\]
A Bell Gram-matrix entry compares two relative products and involves the group element
\[
  \operatorname{rel}(s',t')^{-1}
  \operatorname{rel}(s,t)
  =
  (g_{s'}g_{t'}^{-1})^{-1}g_sg_t^{-1}.
\]
On each free-group factor $\mathbb F_k^{(\alpha,\beta)}$, this word
has reduced length at most four.  It is therefore sufficient to require
the finite representation to agree with the regular trace on all such
words:
\begin{equation}\label{eq:Bell-moment-interface}
  \tau_n(\pi(x))
  =
  \mathbf1_{\{x=e\}}
  \quad
  \text{for every }
  x=(x_{\alpha,\beta})\in G
  \text{ satisfying }
  |x_{\alpha,\beta}|\le4
  \text{ for all }(\alpha,\beta),
\end{equation}
where $|x_{\alpha,\beta}|$ denotes reduced word length in
$\mathbb F_k^{(\alpha,\beta)}$.  The high-girth realization constructed
in Section~\ref{sec:finite} will satisfy a stronger finite-moment
condition and hence, in particular, \eqref{eq:Bell-moment-interface}.

\begin{proposition}[Exact unfiltered Bell profile]
\label{prop:exact-Bell-profile}
Suppose that the representation $\pi:G\to O(n)$ satisfies
\eqref{eq:Bell-moment-interface}.  Then
\begin{equation}\label{eq:exact-Bell}
  \lambda(Z_0)
  =
  v_{\gamma,K}^{\otimes r},
  \qquad
  \rank Z_0
  =
  M_{\gamma,K}^r.
\end{equation}
\end{proposition}

\begin{proof}
By \eqref{eq:vectorization}, for
$s,t,s',t'\in\mathcal W$,
\begin{align}
  [Z_0]_{(s,t),(s',t')}
  &=
  D^{-2}
  \tau_n\!\left(
    (W_{s'}W_{t'}^T)^*
    W_sW_t^T
  \right)
  \notag\\
  &=
  D^{-2}
  \tau_n\!\left(
    \pi\!\left(
      \operatorname{rel}(s',t')^{-1}
      \operatorname{rel}(s,t)
    \right)
  \right).
  \label{eq:finite-gram}
\end{align}
The group element in the last line has word length at most four on
each factor, so \eqref{eq:Bell-moment-interface} gives
\begin{equation}\label{eq:collision-gram}
  [Z_0]_{(s,t),(s',t')}
  =
  \frac1{D^2}
  \mathbf1_{\{
    \operatorname{rel}(s,t)
    =
    \operatorname{rel}(s',t')
  \}}.
\end{equation}
After ordering the branch pairs by their relative-product fibres,
$Z_0$ is therefore a direct sum of all-ones blocks.  A fibre of size
$m$ contributes exactly one nonzero eigenvalue $m/D^2$.

For one block these eigenvalues are precisely the entries of
$v_{\gamma,K}$ in \eqref{eq:base-profile}.  Since
\eqref{eq:relative-product} factorizes over the $r$ blocks, the fibre
multiplicities multiply, and hence the spectrum tensorizes:
\[
  \lambda(Z_0)=v_{\gamma,K}^{\otimes r}.
\]
The number of nonzero eigenvalues is consequently
$M_{\gamma,K}^r$.
\end{proof}

The exact spectrum \eqref{eq:exact-Bell} will be used in different
ways at different R\'enyi orders.  Its rank gives the support bound
needed at low orders, its von Neumann entropy can be computed exactly
at $p=1$, and a simpler vector majorized by the exact profile provides an explicit upper bound at high orders.  At the von Neumann order,
\begin{equation}\label{eq:Bell-von-Neumann}
  H_1(v_{\gamma,K})
  =
  2\log K
  -\frac{\log K}{K}
  -\frac{\gamma-1}{\gamma}\log2.
\end{equation}
For the high-order estimates it is convenient to replace the exact
profile by the simpler vector
\begin{equation}\label{eq:relative-profile}
  w
  =
  \left(
    \frac1K,\,
    \left(\frac1{K^2}\right)^{\!\times(K^2-K)},\,
    0^{\times(K-1)}
  \right).
\end{equation}
Indeed, replacing each pair
\[
  \left(\frac2{K^2},0\right)
  \quad\text{in }v_{\gamma,K}
\]
by
\[
  \left(\frac1{K^2},\frac1{K^2}\right)
\]
is a doubly stochastic operation.  Hence
\begin{equation}\label{eq:profile-majorization}
  w\prec v_{\gamma,K},
  \qquad
  w^{\otimes r}
  \prec
  v_{\gamma,K}^{\otimes r}
  =
  \lambda(Z_0).
\end{equation}
Moreover, for $u\ge0$,
\begin{equation}\label{eq:relative-entropy}
  H_{2+u}(w)
  =
  2\log K
  -
  \frac{\log(K^u+1-1/K)}{1+u}.
\end{equation}
Thus \eqref{eq:exact-Bell} provides the rank control used at low
orders, \eqref{eq:Bell-von-Neumann} gives the exact order-one entropy,
and \eqref{eq:profile-majorization}--\eqref{eq:relative-entropy}
provide the simpler comparison used at high orders.

The actual channel uses $F$ rather than $I_n$, so its Bell output
need not retain the exact spectrum in \eqref{eq:exact-Bell}.  The
following estimate shows that only the average mass removed by the
filter is needed to control this perturbation.

\begin{lemma}[Average filter loss and Bell distance]
\label{lem:Bell-stability}
For any real orthogonal branch tuple and any real symmetric
$F\in M_n(\R)$ satisfying
\[
  0\preceq F\preceq I_n,
\]
one has
\begin{equation}\label{eq:Bell-distance}
  \frac12\|Z_F-Z_0\|_1
  \le
  (D+1)\tau_n(I_n-F^2).
\end{equation}
\end{lemma}

\begin{proof}
Set
\[
  E:=I_n-F^2,
  \qquad
  e:=\tau_n(E),
  \qquad
  f:=\tau_n(E^2),
\]
and define
\[
  \xi_A:=n^{-1/2}\vecop(A),
  \qquad
  \Gamma(A,B)
  :=
  \Phi_0^{\otimes2}
  \bigl(\ket{\xi_A}\bra{\xi_B}\bigr).
\]
Then
\[
  \Gamma(A,A)\succeq0,
  \qquad
  \Tr\Gamma(A,A)=\tau_n(A^*A).
\]
By \eqref{eq:vectorization}, the contribution in which both copies
pass through the filtered part of the channel is
$\Gamma(I_n-E,I_n-E)$.  The remaining contributions form a positive
matrix $Y$ with
\[
  \Tr Y=2e-f.
\]
Consequently,
\[
  Z_F-Z_0
  =
  \Gamma(E,E)
  -\Gamma(E,I_n)
  -\Gamma(I_n,E)
  +Y.
\]
Each entry of $\Gamma(E,I_n)$ is $D^{-2}$ times the normalized trace
of $E$ multiplied by a real orthogonal matrix, and therefore has
modulus at most $e/D^2$.  Hence
\[
  \|\Gamma(E,I_n)\|_1
  \le
  D\|\Gamma(E,I_n)\|_{\HS}
  \le
  De.
\]
Since
\[
  \|\Gamma(E,E)\|_1=f,
  \qquad
  \|Y\|_1=2e-f,
\]
the triangle inequality gives
\[
  \|Z_F-Z_0\|_1
  \le
  f+2De+(2e-f)
  =
  2(D+1)e,
\]
which is \eqref{eq:Bell-distance}.
\end{proof}

Thus the two-copy analysis separates into two interfaces: the
four-letter trace condition determines the unfiltered Bell spectrum
exactly, while the scalar quantity
$\tau_n(I_n-F^2)$ controls the effect of inserting the filter.

\subsection{Spectral perturbation estimates}

The preceding subsection determines the spectrum of the unfiltered
Bell output $Z_0$ exactly and controls the trace distance between
$Z_F$ and $Z_0$.  We now record the two perturbation estimates needed
to transfer this spectral information to $Z_F$.

At low R\'enyi orders, a trace-distance perturbation may create many
small eigenvalues, so the rank of $Z_0$ itself is not stable.  What is
stable is the amount of mass outside its support.  The following
estimate converts this observation into an entropy bound.  It is the
pinching argument used in
\cite[proof of Proposition~B.4, Eqs.~(B.34)--(B.37)]{ShouGorshkov2026},
written here for a general support dimension.

\begin{lemma}[Support-tail bound]
\label{lem:support-tail}
Let $\rho,\sigma\in M_b(\C)$ be states, and let $1\le M<b$ be an
integer.  Suppose
\[
  \rank\sigma\le M,
  \qquad
  \frac12\|\rho-\sigma\|_1\le\delta.
\]
Then, for every $0<\theta\le1/2$,
\begin{equation}\label{eq:support-tail}
  H_\theta(\rho)
  \le
  \log M
  +
  2\frac{b}{M}\delta^\theta.
\end{equation}
\end{lemma}

\begin{proof}
Let $P$ be a rank-$M$ projection whose range contains
$\operatorname{supp}\sigma$.  By \eqref{eq:trace-distance},
\[
  \varepsilon
  :=
  \Tr((I-P)\rho)
  =
  \Tr((I-P)(\rho-\sigma))
  \le\delta.
\]
Pinching with respect to $P\oplus(I-P)$, followed by unitary averaging
within the two blocks, maps $\rho$ to
\[
  \rho'
  =
  (1-\varepsilon)\frac{P}{M}
  +
  \varepsilon\frac{I-P}{b-M}.
\]
By \eqref{eq:unitary-average-majorization} and the Schur concavity of
R\'enyi entropy,
\[
  H_\theta(\rho)\le H_\theta(\rho').
\]
The spectrum of $\rho'$ gives
\[
  H_\theta(\rho')
  =
  \log M
  +
  \frac1{1-\theta}
  \log\left[
    (1-\varepsilon)^\theta
    +
    \left(\frac{b-M}{M}\right)^{1-\theta}
    \varepsilon^\theta
  \right].
\]
Since $0<\theta\le1/2$,
\[
  \frac1{1-\theta}\le2,
  \qquad
  \left(\frac{b-M}{M}\right)^{1-\theta}
  \le\frac bM,
\]
and $\log(1+x)\le x$.  Using $\varepsilon\le\delta$ yields
\eqref{eq:support-tail}.
\end{proof}

For the higher orders we use more detailed information: the unfiltered spectrum majorizes an explicit reference profile.  The next lemma shows that such a majorization relation survives a sufficiently small trace-distance perturbation after adding a uniform spectral floor.

\begin{lemma}[Smoothed majorization]
\label{lem:smoothing}
Let $\rho,\sigma\in M_b(\C)$ be states, and let
$x=(x_1,\ldots,x_b)$ be a decreasing probability vector satisfying
\[
  x_b=0,
  \qquad
  x_1\ge\frac2b,
  \qquad
  x\prec\lambda(\sigma).
\]
Suppose
\[
  \frac12\|\rho-\sigma\|_1\le\delta<\frac1b,
\]
and set
\[
  \varepsilon:=b\delta,
  \qquad
  x_\delta
  :=
  (1-\varepsilon)x+\frac{\varepsilon}{b}\mathbf1_b
  =
  (1-b\delta)x+\delta\mathbf1_b,
\]
where $\mathbf1_b=(1,\ldots,1)$.  Then
\begin{equation}\label{eq:smoothing}
  x_\delta\prec\lambda(\rho).
\end{equation}
Consequently, for $1<p<\infty$,
\begin{equation}\label{eq:renyi-smoothing}
  H_p(\rho)
  \le
  H_p(x)
  -
  \frac{p}{p-1}\log(1-\varepsilon),
\end{equation}
while
\begin{equation}\label{eq:infty-smoothing}
  H_\infty(\rho)
  \le
  H_\infty(x)-\log(1-\varepsilon).
\end{equation}
At the von Neumann order,
\begin{equation}\label{eq:von-Neumann-perturbation}
  H_1(\rho)
  \le
  H_1(x)
  +
  \varepsilon\log b
  +
  h_{\mathrm b}(\varepsilon),
\end{equation}
where
\[
  h_{\mathrm b}(t)
  :=
  -t\log t-(1-t)\log(1-t).
\]
\end{lemma}

\begin{proof}
For a decreasing probability vector $y$, write
\[
  S_j(y):=\sum_{i=1}^j y_i.
\]
Since $x$ is decreasing, the sequence
\[
  S_j(x)-\frac jb
\]
is concave in $j$.  Moreover,
\[
  S_1(x)-\frac1b\ge\frac1b,
  \qquad
  S_{b-1}(x)-\frac{b-1}{b}
  =
  \frac1b,
\]
where the second equality uses $x_b=0$.  Hence
\[
  S_j(x)-\frac jb\ge\frac1b,
  \qquad 1\le j<b.
\]
It follows that
\[
  S_j(x_\delta)
  =
  S_j(x)
  -
  b\delta
  \left(
    S_j(x)-\frac jb
  \right)
  \le
  S_j(x)-\delta.
\]
Using $x\prec\lambda(\sigma)$ and
\eqref{eq:ky-fan-trace-perturbation},
\[
  S_j(x_\delta)
  \le
  S_j(\lambda(\sigma))-\delta
  \le
  S_j(\lambda(\rho)),
  \qquad 1\le j<b.
\]
The total masses agree, proving \eqref{eq:smoothing}.

Since
\[
  x_\delta\ge(1-\varepsilon)x
\]
coordinatewise, Schur concavity together with
\eqref{eq:smoothing} gives, for $1<p<\infty$,
\[
  H_p(\rho)
  \le
  H_p(x_\delta)
  \le
  H_p(x)
  -
  \frac{p}{p-1}\log(1-\varepsilon),
\]
which is \eqref{eq:renyi-smoothing}.  The same argument using the
largest coordinate gives \eqref{eq:infty-smoothing}.

Finally,
\[
  x_\delta
  =
  (1-\varepsilon)x
  +
  \varepsilon\frac{\mathbf1_b}{b},
\]
so the classical mixture-entropy bound gives
\[
  H_1(x_\delta)
  \le
  (1-\varepsilon)H_1(x)
  +
  \varepsilon\log b
  +
  h_{\mathrm b}(\varepsilon).
\]
Since $x_\delta\prec\lambda(\rho)$,
$H_1(\rho)\le H_1(x_\delta)$, and
\eqref{eq:von-Neumann-perturbation} follows.
\end{proof}

\subsection{Nonadditivity certificates}

We first specialize the preceding estimates to the parameters used in
Theorem~\ref{thm:main}:
\begin{equation}\label{eq:base-parameters}
 k=2048,\qquad\gamma=8,\qquad K = 2^{15} = 32768,
 \qquad D=K^r.
\end{equation}
Set
\begin{equation}\label{eq:compression-targets}
 C_r^2=(K+1000)^r-K^r,\qquad
 B_r=34^r-1,\qquad a=\frac{K+1000}{K}=\frac{4221}{4096}.
\end{equation}

\begin{theorem}[Uniform positive-order certificate]\label{thm:certificate}
Let $p_0>0$, $r\ge1$, and $\theta=\min\{p_0,1/2\}$. Suppose a representation of \eqref{eq:product-model}, with \eqref{eq:base-parameters}, satisfies \eqref{eq:Bell-moment-interface}. Let $F$ be real symmetric with $0\preceq F\preceq I_n$, and assume \eqref{eq:compression-interface} with $C=C_r$, $B=B_r$. If there exists $\delta>0$ such that
\begin{equation}\label{eq:cost-interface}
 (D+1)\tau_n(I_n-F^2)\le\delta,\qquad
 D^2\delta\le2^{-16},\qquad
 \delta^\theta\le2^{-(r+16)},
\end{equation}
then
\begin{equation}\label{eq:certified-gap}
 \inf_{p\in[p_0,\infty]}\Delta_p(\Phi_F;\Omega_n)
 \ge\frac r2+\frac{16r-1}{16384}>\frac r2.
\end{equation}
\end{theorem}
\begin{proof}
By \cref{prop:one-copy}, the one-copy entropy bound \eqref{eq:one-copy-entropy} holds with $a_0=a^r$ and $b_0=34^r$. On the two-copy side, the four-letter trace condition
\eqref{eq:Bell-moment-interface} identifies the Gram matrix of the
unfiltered Bell output with the relative-product collision matrix.
Hence, by \eqref{eq:exact-Bell},
\[
  \lambda(Z_0)=v_{\gamma,K}^{\otimes r},
  \qquad
  \rank Z_0=M_{\gamma,K}^r.
\]
Moreover, \cref{lem:Bell-stability} together with
\eqref{eq:cost-interface} gives
\[
  \frac12\|Z_F-Z_0\|_1
  \le
  (D+1)\tau_n(I_n-F^2)
  \le
  \delta.
\]

For $p_0\le p\le2$, use \cref{lem:support-tail} with $b=D^2$ and $M=(9K^2/16)^r$. By \eqref{eq:cost-interface},
\begin{equation}\label{eq:low-order-gap}
 \Delta_p(\Phi_F;\Omega_n)
 \ge r\bigl(\log(16/9)-2\log a\bigr)
       -2(16/9)^r\delta^\theta
 \ge r c_0-2^{-15},
\end{equation}
where $c_0=\log(16/9)-2\log a$.

For $p=2+u$, $u\ge0$, apply \cref{lem:smoothing} with $x$ the decreasing rearrangement of $w^{\otimes r}$, and put $\zeta=D^2\delta$. Since
$x_\delta\ge(1-\zeta)x$ coordinatewise and $H_p(x)=rH_p(w)$,
\[
 H_p(Z_F)\le rH_p(w)-\frac p{p-1}\log(1-\zeta)
 \le rH_p(w)-2\log(1-\zeta).
\]
The weighted arithmetic--geometric mean inequality gives
\[
 \log(K^u+1-1/K)
 \ge h_{\mathrm b}(3/4)+\tfrac34u\log K+\tfrac14\log(1-1/K).
\]
Combining this with \eqref{eq:relative-entropy} and \eqref{eq:one-copy-entropy},
\begin{equation}\label{eq:high-order-gap}
 \Delta_{2+u}(\Phi_F;\Omega_n)
 \ge r\frac{c_1+u c_2}{1+u}+2\log(1-\zeta)
 \ge r\min\{c_1,c_2\}-2^{-14},
\end{equation}
where
\begin{equation}\label{eq:comparison-constants}
 c_1=h_{\mathrm b}(3/4)+\tfrac14\log(1-1/K)-2\log a,
 \qquad c_2=\tfrac34\log K-2\log34.
\end{equation}
Here $\log(1-x)\ge-2x$ for $0\le x\le1/2$ was used. At $p=\infty$, the same smoothing bound gives
$\Delta_\infty\ge r(\log K-2\log34)+\log(1-\zeta)$,
which is at least $r c_2-2^{-14}$.

For completeness, the required constants have rational certificates. The expansion
\begin{equation}\label{eq:log-series}
 \log x=2\sum_{j=0}^\infty\frac{z^{2j+1}}{2j+1},\qquad
 z=\frac{x-1}{x+1},
\end{equation}
gives
\begin{equation}\label{eq:log-certificates}
 \log2>\frac{53056}{76545}>\frac{6931}{10000},\qquad
 \log3<\frac{33223}{30240}<\frac{10987}{10000}.
\end{equation}
For the first bound take four terms at $z=1/3$; for the second take four terms at $z=1/2$ and bound the remainder by $2z^9/[9(1-z^2)]$. With $\log a\le125/4096$ and $\tfrac14\log(1-1/K)\ge-1/65536$,
\begin{align}
 c_0&>\frac{23}{40}-\frac{125}{2048}
      =\frac{5263}{10240}>\frac12+\frac1{1024},\label{eq:c0-certificate}\\
 c_1&>\frac{22487}{40000}-\frac{125}{2048}-\frac1{65536}
      =\frac{20526063}{40960000}>\frac12+\frac1{1024},\label{eq:c1-certificate}\\
 c_2&=\tfrac54\log2-2\log(17/16)
      >\frac56-\frac18=\frac{17}{24}>\frac12+\frac1{1024}.\label{eq:c2-certificate}
\end{align}
Equations \eqref{eq:low-order-gap}--\eqref{eq:high-order-gap} now give \eqref{eq:certified-gap}.
\end{proof}

For the separate $p=1$ specialization, the purity estimate from \cref{prop:one-copy} already provides the required one-copy control, so the linear test $L_z$ and its associated compression bound are not needed.  Moreover, since only a fixed positive gap is required in \cref{cor:von-neumann}, a single block $r=1$ already suffices and no outer amplification is necessary. We record the specialization used for \cref{cor:von-neumann}.

\begin{proposition}[A von Neumann certificate]
\label{prop:von-Neumann-certificate}
Let $k\ge1$ and $\gamma\ge2$, and specialize
\eqref{eq:product-model} to $r=1$.  Thus
\[
  G=\mathbb F_k^\gamma,
  \qquad
  K=2\gamma k,
  \qquad
  D=K.
\]
Let
\[
  \pi:G\to O(n)
\]
be a real orthogonal representation satisfying the four-letter trace
condition \eqref{eq:Bell-moment-interface}, and let
$F\in M_n(\R)$ satisfy
\[
  F=F^T,
  \qquad
  0\preceq F\preceq I_n.
\]
Suppose there exist constants $C\ge0$ and
$0<\delta<K^{-2}$ such that
\begin{equation}\label{eq:von-Neumann-interface}
  \|FT_AF\|
  \le
  C\|A\|_{\HS}
  \quad
  \text{for every }A\in\mathcal A_K,
  \qquad
  (K+1)\tau_n(I_n-F^2)
  \le
  \delta.
\end{equation}
Set
\[
  \varepsilon:=K^2\delta.
\]
Then
\begin{equation}\label{eq:von-Neumann-certificate}
  \Delta_1(\Phi_F;\Omega_n)
  \ge
  \frac{\gamma-1}{\gamma}\log2
  +
  \frac{\log K}{K}
  -
  2\log\left(1+\frac{C^2}{K}\right)
  -
  2\varepsilon\log K
  -
  h_{\mathrm b}(\varepsilon).
\end{equation}
\end{proposition}

\begin{proof}
We estimate the one-copy and two-copy terms separately. By the purity part of \cref{prop:one-copy}, the first inequality in
\eqref{eq:von-Neumann-interface} gives
\begin{equation}\label{eq:von-Neumann-one-copy}
  H_{\min,1}(\Phi_F)
  \ge
  \log K
  -
  \log\left(1+\frac{C^2}{K}\right).
\end{equation}

For the two-copy term, let
\[
  Z_0=\Phi_0^{\otimes2}(\Omega_n),
  \qquad
  Z_F=\Phi_F^{\otimes2}(\Omega_n).
\]
Since $\pi$ satisfies \eqref{eq:Bell-moment-interface},
\eqref{eq:exact-Bell} with $r=1$ gives
\[
  \lambda(Z_0)=v_{\gamma,K}.
\]
Moreover, by \cref{lem:Bell-stability} and the second inequality in
\eqref{eq:von-Neumann-interface},
\begin{equation}\label{eq:von-Neumann-Bell-distance}
  \frac12\|Z_F-Z_0\|_1
  \le
  (K+1)\tau_n(I_n-F^2)
  \le
  \delta.
\end{equation}

Let $x$ be the decreasing rearrangement of $v_{\gamma,K}$, viewed as
a probability vector of length $K^2$.  Its largest entry is
\[
  x_1=\frac1K\ge\frac2{K^2},
\]
and, since $M_{\gamma,K}<K^2$, its last entry satisfies
$x_{K^2}=0$.  Also $\delta<K^{-2}$.  Hence
\cref{lem:smoothing} applies with
\[
  b=K^2,
  \qquad
  \rho=Z_F,
  \qquad
  \sigma=Z_0,
  \qquad
  \varepsilon=b\delta=K^2\delta.
\]
Using \eqref{eq:von-Neumann-perturbation},
\begin{equation}\label{eq:von-Neumann-two-copy}
  H_1(Z_F)
  \le
  H_1(v_{\gamma,K})
  +
  2\varepsilon\log K
  +
  h_{\mathrm b}(\varepsilon).
\end{equation}
The exact one-block Bell entropy is
\begin{equation*}
  H_1(v_{\gamma,K})
  =
  2\log K
  -
  \frac{\log K}{K}
  -
  \frac{\gamma-1}{\gamma}\log2
\end{equation*}
by \eqref{eq:Bell-von-Neumann}.  Combining this with
\eqref{eq:von-Neumann-one-copy} and
\eqref{eq:von-Neumann-two-copy} in
\[
  \Delta_1(\Phi_F;\Omega_n)
  =
  2H_{\min,1}(\Phi_F)-H_1(Z_F)
\]
gives \eqref{eq:von-Neumann-certificate}.
\end{proof}

\section{A high-girth finite realization}\label{sec:finite}

We construct the branch matrices needed in \Cref{thm:main,cor:von-neumann} and verify their finite-moment identities. Throughout, $\mathbb K$ is the ordered field of $256$ elements fixed in Section~\ref{sec:fields}. The polynomial degree parameter $m$ will be specified in Section~\ref{sec:filter}.

\subsection{An explicitly indexed coordinate fibre}
For an integer $\ell\ge2$, specialize the coordinate graph and its zero fibre from Section~\ref{sec:graph-background} to
\begin{equation}\label{eq:graph-parameters}
 \nu=4\ell-3,\qquad b=\ell-1,\qquad
 \mathcal V_\ell
 =\{v\in V(D(4\ell-3,256)):a_j(v)=0\ (2\le j\le\ell-1)\}.
\end{equation}
For $\ell=2$, the list of constraints is empty, and hence $\mathcal V_2=V(D(5,256))$. Let $\Gamma_\ell$ be the induced graph on $\mathcal V_\ell$. We order the vertices canonically by placing all point vertices before all line vertices and, within each class, ordering the tuples of free coordinates lexicographically according to the coordinate order \eqref{eq:graph-coordinates} and the field order
$c_0<\cdots<c_{255}$.

\begin{proposition}[Size, girth, and prescribed indexing]
\label{prop:fibre}
For every integer $\ell\ge2$, let $\Gamma_\ell$ be the subgraph of
$D(4\ell-3,256)$ induced by $\mathcal V_\ell$.  Then, for every
$v\in\mathcal V_\ell$ and every colour $c\in\mathbb F_{256}$, the
unique colour-$c$ neighbour of $v$ in $D(4\ell-3,256)$ also belongs
to $\mathcal V_\ell$.  Consequently, $\Gamma_\ell$ is
$256$-regular, and every vertex of $\Gamma_\ell$ is incident to
exactly one edge of each of the $256$ colours.

Moreover,
\[
  \operatorname{girth}\Gamma_\ell\ge4\ell+2,
\]
and its vertex set admits the explicit indexing
\begin{equation}\label{eq:leg-dimension}
  \mathcal V_\ell
  \cong
  \{\mathrm P,\mathrm L\}\times\mathbb K^{3\ell-1},
  \qquad
  q_\ell:=|\mathcal V_\ell|
  =
  2\cdot256^{3\ell-1}
  =
  2^{24\ell-7}.
\end{equation}
\end{proposition}

\begin{proof}
Here
\[
  \left\lfloor\frac{\nu+2}{4}\right\rfloor
  =
  \ell-1.
\]
If $v\in\mathcal V_\ell$ and $w\sim v$ in
$D(4\ell-3,256)$, then \eqref{eq:adjacency-invariance} gives
\[
  a_j(w)=a_j(v)=0,
  \qquad
  2\le j\le\ell-1.
\]
Hence $w\in\mathcal V_\ell$.  Since every vertex of
$D(4\ell-3,256)$ has a unique neighbour of each colour
$c\in\mathbb F_{256}$, the same is true in the induced graph
$\Gamma_\ell$.  Thus $\Gamma_\ell$ is $256$-regular.

By \eqref{eq:general-fibre-recursion}--\eqref{eq:general-fibre-count},
the constraints leave
\[
  \nu-(\ell-1)+1=3\ell-1
\]
free coordinates in each of the point and line classes, while all
remaining coordinates are determined successively.  Hence
\[
  \mathcal V_\ell
  \cong
  \{\mathrm P,\mathrm L\}\times\mathbb K^{3\ell-1},
\]
which gives \eqref{eq:leg-dimension}.  Under this identification, the
vertex ordering fixed above is obtained by placing $\mathrm P$ before
$\mathrm L$ and ordering the free-coordinate tuples lexicographically
with respect to the coordinate order \eqref{eq:graph-coordinates} and
the field order $c_0<\cdots<c_{255}$.  Thus the indexing is completely
prescribed.

Finally, every cycle in $\Gamma_\ell$ is also a cycle in
$D(4\ell-3,256)$, so
\[
  \operatorname{girth}\Gamma_\ell
  \ge
  \operatorname{girth}D(4\ell-3,256).
\]
By \cref{prop:LUW},
\[
  \operatorname{girth}D(4\ell-3,256)
  \ge
  (4\ell-3)+5
  =
  4\ell+2.
\]
\end{proof}

Pair the colours as $(c_0,c_1),\ldots,(c_{254},c_{255})$. Since the fibre is invariant under the paired-colour permutations in \eqref{eq:general-paired-rule}, their restrictions define
\begin{equation}\label{eq:paired-permutations}
 P_i^{(\ell)}e_v=e_{u_i(v)},\qquad
 1\le i\le128,\quad v\in\mathcal V_\ell.
\end{equation}
The neighbour rules \eqref{eq:incidence} and \eqref{eq:edge-colour}, together with the fibre recursion, specify every matrix entry. For $w\in\mathbb F_{128}$, write $P_w^{(\ell)}$ for its evaluation at these matrices, as in \eqref{eq:word-walk-convention}.

\begin{lemma}[Raw short-word traces]\label{lem:graph-trace}
For every integer $\ell\ge2$ and every $w\in\mathbb F_{128}$ with $|w|\le4\ell$,
\begin{equation}\label{eq:raw-traces}
 \tau_{q_\ell}(P_w^{(\ell)})=\mathbf1_{\{w=e\}}.
\end{equation}
\end{lemma}
\begin{proof}
Apply \cref{lem:coloured-short-traces} to the full graph $D(4\ell-3,256)$. For $w\ne e$ in the stated range, the resulting permutation has no fixed vertex. Its restriction to the invariant set $\mathcal V_\ell$ therefore also has trace zero. For $w=e$, the restricted matrix is $I_{q_\ell}$.
\end{proof}

\subsection{The two generator alphabets}

The main construction needs $2048$ generators on each free factor, whereas \eqref{eq:paired-permutations} supplies $128$. We therefore obtain the required $2048$-generator family through a short free-subgroup encoding.  In the smaller von Neumann construction only $64$ generators are needed, so a subset of the raw generators can be used directly.

\begin{lemma}[A length-three free encoding]
\label{lem:free-encoding}
Let
\[
  A=\mathbb F_{64}
   =\langle a_1,\ldots,a_{64}\rangle,
  \qquad
  B=\mathbb F_{32}
   =\langle b_1,\ldots,b_{32}\rangle,
\]
and regard
\[
  A*B
  =
  \langle
    a_1,\ldots,a_{64},
    b_1,\ldots,b_{32}
  \rangle
  \cong\mathbb F_{96}.
\]
For
\[
  1\le i\le64,
  \qquad
  1\le j\le32,
\]
define
\begin{equation}\label{eq:abstract-encoding}
  c_{ij}:=a_i b_j a_i^{-1}.
\end{equation}
Then the $2048$ elements $c_{ij}$ freely generate a subgroup of
$A*B$ isomorphic to $\mathbb F_{2048}$.  Moreover, every reduced word
of length $L$ in the letters $c_{ij}^{\pm1}$ expands to a word of
reduced length at most $3L$ in the free basis
\[
  \{a_1,\ldots,a_{64},b_1,\ldots,b_{32}\}.
\]
\end{lemma}

\begin{proof}
Let
\[
  w=c_{i_1j_1}^{\varepsilon_1}\cdots
    c_{i_Lj_L}^{\varepsilon_L},
  \qquad
  \varepsilon_\nu\in\{\pm1\},
\]
be a nonempty reduced word in the $c_{ij}$.
Group consecutive letters having the same first index $i$ into
maximal runs.  A run with first index $i$ has the form
\[
  c_{ij_1}^{\varepsilon_1}\cdots
  c_{ij_t}^{\varepsilon_t}
  =
  a_i
  \bigl(
    b_{j_1}^{\varepsilon_1}\cdots
    b_{j_t}^{\varepsilon_t}
  \bigr)
  a_i^{-1}.
\]
Because the original word is reduced, the word between parentheses is
a nonempty reduced word in the free basis of $B$. Between two consecutive maximal runs the first indices are distinct.
Hence the intervening factor
\[
  a_i^{-1}a_{i'}
  \qquad(i\neq i')
\]
is a nontrivial reduced word in $A$.  The expansion of $w$ is
therefore a nonempty reduced alternating product of nontrivial
elements of $A$ and $B$.  By the normal form in the free product
$A*B$, it cannot represent the identity.  Thus the $c_{ij}$ freely
generate a free subgroup of rank
\[
  64\cdot32=2048.
\]

Finally, each letter $c_{ij}^{\pm1}$ is represented by three letters in the original free basis.  Hence a word of length $L$ expands to a word of length at most $3L$.
\end{proof}

We now implement this encoding using the $128$ permutation matrices
from \eqref{eq:paired-permutations}.  Fix an integer $m\ge2$ and, for
the main construction, take
\[
  \ell=3m,
  \qquad
  q:=q_{3m}.
\]
Using the ordering of the paired-colour permutations induced by the
field order $c_0,\ldots,c_{255}$ fixed in Section~\ref{sec:fields},
identify the abstract generators above with the first $96$ raw
generators by
\[
  a_i\longleftrightarrow P_i^{(3m)},
  \qquad 1\le i\le64,
\]
and
\[
  b_j\longleftrightarrow P_{64+j}^{(3m)},
  \qquad 1\le j\le32.
\]
Accordingly, define the $2048$ encoded permutation matrices
\begin{equation}\label{eq:encoded-permutations}
  U_{32(i-1)+j}
  :=
  P_i^{(3m)}
  P_{64+j}^{(3m)}
  \bigl(P_i^{(3m)}\bigr)^{-1},
  \qquad
  1\le i\le64,\quad
  1\le j\le32.
\end{equation}
These are the matrices that represent the $2048$ generators required
on each free factor in the main construction.

To verify the required short-word behaviour, let
$w\in\mathbb F_{2048}$ be represented by a reduced word of length at
most $4m$ in the encoded generators.  By
\cref{lem:free-encoding}, if $w\neq e$, its expansion in the raw free
basis is still nontrivial and has length at most
\[
  3|w|\le12m.
\]
Since $\ell=3m$, \cref{lem:graph-trace} applies up to raw word length
\[
  4\ell=12m.
\]
Therefore
\begin{equation}\label{eq:encoded-traces}
  \tau_q(U_w)
  =
  \mathbf1_{\{w=e\text{ in }\mathbb F_{2048}\}},
  \qquad
  |w|\le4m.
\end{equation}
Thus the finite family
\[
  U_1,\ldots,U_{2048}\in O(q)
\]
has exactly the regular free-group traces required below throughout
the word-length range $4m$.

For the smaller von Neumann construction only $64$ generators are
needed on each free factor, so no encoding is required.  We use the
first $64$ raw generators directly:
\begin{equation}\label{eq:direct-alphabet}
  \widetilde U_j
  :=
  P_j^{(m)}
  \in O(q_m),
  \qquad
  1\le j\le64.
\end{equation}
Applying \cref{lem:graph-trace} with $\ell=m$ and restricting the raw
alphabet to these generators gives
\begin{equation}\label{eq:direct-traces}
  \tau_{q_m}(\widetilde U_w)
  =
  \mathbf1_{\{w=e\text{ in }\mathbb F_{64}\}},
  \qquad
  |w|\le4m.
\end{equation}

\subsection{Tensor-leg representations and the moment interface}

The preceding subsection provides a finite-dimensional real
permutation-matrix realization for each individual free factor:
the main construction uses
\[
  U_1,\ldots,U_{2048}\in O(q),
  \qquad q=q_{3m},
\]
with the short-word trace identity \eqref{eq:encoded-traces}, while
the smaller von Neumann construction uses
\[
  \widetilde U_1,\ldots,\widetilde U_{64}\in O(q_m)
\]
with \eqref{eq:direct-traces}.  We now place copies of these
representations on separate tensor legs to realize the full
direct-product group $G$ in \eqref{eq:product-model}.  The tensor
structure makes the commutation between distinct free factors exact.

For the main construction, recall from \eqref{eq:base-parameters} that
\[
  \gamma=8,
  \qquad
  k=2048,
  \qquad
  G
  =
  \prod_{\beta=1}^r
  \prod_{\alpha=1}^{8}
  \mathbb F_{2048}^{(\alpha,\beta)}.
\]
Since $\mathbb F_{2048}$ is free, the assignment
\[
  g_j\longmapsto U_j,
  \qquad
  1\le j\le2048,
\]
extends uniquely to an orthogonal representation
\[
  \rho_m:\mathbb F_{2048}\longrightarrow O(q).
\]
For each block $\beta$, set
\[
  \mathcal H_\beta
  :=
  \bigotimes_{\alpha=1}^{8}\C^q
\]
and define
\begin{equation}\label{eq:block-representation}
  \pi_\beta:
  \prod_{\alpha=1}^{8}
  \mathbb F_{2048}^{(\alpha,\beta)}
  \longrightarrow
  O(q^8),
  \qquad
  \pi_\beta\bigl((x_{\alpha,\beta})_{\alpha=1}^{8}\bigr)
  :=
  \bigotimes_{\alpha=1}^{8}\rho_m(x_{\alpha,\beta}).
\end{equation}
Thus the $\alpha$th free factor in block $\beta$ acts only on the
$\alpha$th tensor leg of $\mathcal H_\beta$.  In particular, the
actions of distinct free factors commute exactly.

The full input space is
\[
  \mathcal H
  :=
  \bigotimes_{\beta=1}^r\mathcal H_\beta
  =
  \bigotimes_{\beta=1}^r
  \bigotimes_{\alpha=1}^{8}\C^q,
\]
and the block representations combine to
\begin{equation}\label{eq:finite-product-representation}
  \pi:G\longrightarrow O(n),
  \qquad
  \pi\bigl((x_{\alpha,\beta})_{\alpha,\beta}\bigr)
  :=
  \bigotimes_{\beta=1}^r
  \bigotimes_{\alpha=1}^{8}
  \rho_m(x_{\alpha,\beta}),
  \qquad
  n=q^{8r}=2^{8r(72m-7)}.
\end{equation}

We order the tensor legs first by the block index $\beta$ and then by
the factor index $\alpha$.  On each leg we use the vertex ordering
fixed in \cref{prop:fibre}.  For the branch labels we retain the order from Section~\ref{sec:free}: $\mathcal S$ is ordered
lexicographically by $(\alpha,j,\varepsilon)$, with $+1$ before
$-1$, and $\mathcal W=\mathcal S^r$ carries the induced
lexicographic order.  These conventions fix the input and output
matrix indices completely.

Now let
\[
  s=(s_1,\ldots,s_r)\in\mathcal W,
  \qquad
  s_\beta=g_{j_\beta,\alpha_\beta}^{\varepsilon_\beta}.
\]
The group element $g_s$ defined in
\eqref{eq:branch-group-element} has exactly one nontrivial factor in
each block.  Its finite-dimensional representative is therefore
\begin{equation}\label{eq:finite-branches}
  \begin{aligned}
  W_s
  :=
  \pi(g_s)
  &=
  \bigotimes_{\beta=1}^r
  \bigotimes_{\alpha=1}^{8}
  B_{\alpha,\beta}(s),
  \\
  B_{\alpha,\beta}(s)
  &=
  \begin{cases}
    U_{j_\beta}^{\varepsilon_\beta},
      & \alpha=\alpha_\beta,\\
    I_q,
      & \alpha\neq\alpha_\beta.
  \end{cases}
  \end{aligned}
\end{equation}
These are the branch matrices entering the channel
$\Phi_F$ in \eqref{eq:channel}.

For the smaller von Neumann construction, take
\[
  r=1,
  \qquad
  \gamma=2,
  \qquad
  k=64.
\]
The assignment
\[
  g_j\longmapsto\widetilde U_j,
  \qquad 1\le j\le64,
\]
defines
\[
  \widetilde\rho_m:
  \mathbb F_{64}\longrightarrow O(q_m).
\]
Hence
\[
  \widetilde\pi:
  \mathbb F_{64}^{(1)}
  \times
  \mathbb F_{64}^{(2)}
  \longrightarrow
  O(q_m^2),
  \qquad
  \widetilde\pi(x_1,x_2)
  =
  \widetilde\rho_m(x_1)
  \otimes
  \widetilde\rho_m(x_2).
\]
The corresponding $K=256$ branch matrices are
\begin{equation}\label{eq:small-branches}
  \begin{aligned}
  \widetilde W_{(1,j,\varepsilon)}
  &=
  \widetilde U_j^\varepsilon\otimes I_{q_m},
  \\
  \widetilde W_{(2,j,\varepsilon)}
  &=
  I_{q_m}\otimes\widetilde U_j^\varepsilon,
  \qquad
  1\le j\le64,\quad
  \varepsilon\in\{+1,-1\}.
  \end{aligned}
\end{equation}
Its input dimension is
\[
  q_m^2=2^{48m-14}.
\]

\begin{proposition}[Exact finite moments]
\label{prop:finite-moments}
Let $m\ge2$, and let
\[
  \pi:G\to O(n)
\]
be either of the finite-dimensional representations constructed
above.  Suppose
\[
  x=(x_{\alpha,\beta})\in G
\]
satisfies
\[
  |x_{\alpha,\beta}|\le4m
  \qquad
  \text{for every tensor leg }(\alpha,\beta).
\]
Then
\begin{equation}\label{eq:product-trace}
  \tau_n(\pi(x))
  =
  \tau_G(\lambda_G(x))
  =
  \mathbf1_{\{x=e\}}.
\end{equation}
In particular, the four-letter condition
\eqref{eq:Bell-moment-interface} holds.

Moreover, for every $A\in\mathcal A_D$ and $z\in\C^D$, the finite
tests in \eqref{eq:finite-tests} satisfy
\begin{equation}\label{eq:linear-moment-transfer}
  \begin{aligned}
  \tau_n(T_A^j)
  &=
  \tau_G\bigl((T_A^{\mathrm f})^j\bigr),
  &&0\le j\le2m,
  \\
  \tau_n\bigl((L_z^*L_z)^m\bigr)
  &=
  \tau_G\!\left(
    \bigl((L_z^{\mathrm f})^*L_z^{\mathrm f}\bigr)^m
  \right).
  \end{aligned}
\end{equation}
\end{proposition}

\begin{proof}
Consider first the main construction.  By
\eqref{eq:finite-product-representation}, normalized traces factor
over the tensor legs:
\[
  \tau_n(\pi(x))
  =
  \prod_{\beta=1}^r
  \prod_{\alpha=1}^{8}
  \tau_q\bigl(\rho_m(x_{\alpha,\beta})\bigr).
\]
Since $|x_{\alpha,\beta}|\le4m$ on every leg,
\eqref{eq:encoded-traces} gives
\[
  \tau_q\bigl(\rho_m(x_{\alpha,\beta})\bigr)
  =
  \mathbf1_{\{x_{\alpha,\beta}=e\}}.
\]
Therefore
\[
  \tau_n(\pi(x))
  =
  \prod_{\alpha,\beta}
  \mathbf1_{\{x_{\alpha,\beta}=e\}}
  =
  \mathbf1_{\{x=e\}},
\]
which is the regular trace of $G$.  The same argument for the smaller
construction uses \eqref{eq:direct-traces} on its two tensor legs.
This proves \eqref{eq:product-trace}.  Since $m\ge2$, the same identity
in particular holds whenever every leg has word length at most four,
and hence implies \eqref{eq:Bell-moment-interface}.

It remains to verify \eqref{eq:linear-moment-transfer}.  Each branch
matrix $W_s=\pi(g_s)$ contributes at most one generator on any fixed
tensor leg.  Consequently, a monomial in
\[
  T_A^j
  =
  \left(
    \sum_{s\neq t}A_{st}W_s^*W_t
  \right)^j
\]
contains at most $2j\le4m$ letters on each leg.  Applying
\eqref{eq:product-trace} term by term therefore gives
\[
  \tau_n(T_A^j)
  =
  \tau_G\bigl((T_A^{\mathrm f})^j\bigr),
  \qquad
  0\le j\le2m.
\]
Likewise, every monomial in $(L_z^*L_z)^m$ contains at most $2m$
letters on each leg, so the same argument yields the second identity
in \eqref{eq:linear-moment-transfer}.
\end{proof}

\section{The filter and its parameter estimates}\label{sec:filter}

We construct $F$ and verify the compression and average-loss conditions of \Cref{thm:certificate}. The final subsection verifies the single-penalty specialization required by \cref{prop:von-Neumann-certificate}.

\subsection{The complete parameter prescription}
\label{sec:filter-parameters}

We now fix the parameters used in the construction of
\Cref{thm:main}.  They are chosen so that the error requirements in
\Cref{thm:certificate} are satisfied while keeping the resulting
finite-dimensional realization explicit.

For $p_0>0$ and an integer $r\ge1$, use
\eqref{eq:base-parameters} and set
\begin{equation}\label{eq:parameter-list}
  \begin{aligned}
  D&=2^{15r},
  & d&=D(D-1),
  & \theta&=\min\{p_0,1/2\},
  \\
  s&=\max\left\{
       30r+16,\,
       \left\lceil\frac{r+16}{\theta}\right\rceil
      \right\},
  & \delta&=2^{-s},
  \\
  t&=s+15r+3,
  \\
  h&=\max\left\{
      \left\lceil\frac{4d+t+3}{16}\right\rceil,\,
      6(D-1)+t
     \right\},
  & m&=64h.
  \end{aligned}
\end{equation}

The choice of $s$ immediately gives the two error bounds required in
\eqref{eq:cost-interface}.  Indeed, since $D^2=2^{30r}$,
\begin{equation}\label{eq:error-budget}
  D^2\delta
  =
  2^{30r-s}
  \le
  2^{-16},
  \qquad
  \delta^\theta
  =
  2^{-s\theta}
  \le
  2^{-(r+16)}.
\end{equation}
The remaining quantity $t$ is chosen so that
\begin{equation}\label{eq:trace-budget-reduction}
  (D+1)2^{1-t}
  =
  \frac{D+1}{4D}\,2^{-s}
  <
  \delta.
\end{equation}
Thus it will be enough below to obtain a normalized trace bound of
order $2^{1-t}$.

We will also use the following two numerical margins.  By
\eqref{eq:free-bounds},
\[
  J_r^2=(K+968)^r-K^r,
\]
whereas the targets in \eqref{eq:compression-targets} satisfy
\begin{equation}\label{eq:target-ratios}
  \frac{C_r^2}{J_r^2}
  \ge
  \frac{1000}{968}
  =
  \frac{125}{121},
  \qquad
  \frac{B_r}{32^r}
  \ge
  \frac{33}{32}.
\end{equation}
For the first inequality, compare the positive-degree terms in the
two binomial expansions.  For the second,
\[
  34^r-1
  =
  33\sum_{j=0}^{r-1}34^j
  \ge
  33\cdot32^{r-1}.
\]

The elementary estimates
\begin{equation}\label{eq:growth-certificates}
  \left(\frac{10}{9}\right)^8>2,
  \qquad
  \left(\frac{2079}{2048}\right)^{64}>2
\end{equation}
will also be used below.  The first is immediate from
$10^8>2\cdot9^8$, while
\[
  \left(1+\frac{31}{2048}\right)^{64}
  >
  1+64\frac{31}{2048}
   +\binom{64}{2}\left(\frac{31}{2048}\right)^2
  >2.
\]

Finally, the parameter $m$ fixes the member of the high-girth family
constructed in Section~\ref{sec:finite}.  Since
\[
  q=q_{3m}=2^{72m-7},
\]
the resulting input dimension is
\[
  n=q^{8r}=2^{N(r,p_0)},
\]
where
\begin{equation}\label{eq:input-budget}
  N(r,p_0)
  =
  8r(72m-7)
  =
  36864rh-56r.
\end{equation}
To obtain the simpler bound stated in \Cref{thm:main}, note that
$4\mid D$ and
\[
  \left\lceil\frac{4d+t+3}{16}\right\rceil
  =
  \frac d4+
  \left\lceil\frac{t+3}{16}\right\rceil
  \le
  \frac d4+t,
  \qquad
  6(D-1)\le\frac d4.
\]
Hence
\[
  h\le\frac d4+t,
  \qquad
  N(r,p_0)
  <
  9216rD^2+36864rt.
\]
Moreover,
\[
  s\le30r+17+\frac{r+16}{\theta},
  \qquad
  t\le45r+20+\frac{r+16}{\theta},
\]
which, together with $D^2=2^{30r}$ and
$\theta=\min\{p_0,1/2\}$, gives
\eqref{eq:main-dimension-bound}.

\subsection{A real spherical Chebyshev penalty}

The first penalty enforces the first inequality in \eqref{eq:compression-interface} uniformly over $\mathcal A_D$. Give this real Hilbert space the ordered orthonormal basis
\begin{equation}\label{eq:coefficient-basis}
 \frac{E_{ij}+E_{ji}}{\sqrt2},\qquad
 \frac{\mathrm i(E_{ij}-E_{ji})}{\sqrt2},
 \qquad 1\le i<j\le D,
\end{equation}
ordered by $(i,j)$, with the first matrix before the second. Writing
\[
  A=\sum_{a=1}^d x_aA_a
\]
in this basis gives
\[
  \|A\|_{\HS}^2=\sum_{a=1}^d x_a^2.
\]
Hence the Hilbert--Schmidt unit sphere of $\mathcal A_D$ is
identified with $S_\R^{d-1}$.  We use $\sigma_d$ for the normalized
surface measure on this sphere, equivalently for its pushforward to
$\{A\in\mathcal A_D:\|A\|_{\HS}=1\}$ under
$x\mapsto\sum_a x_aA_a$.

\begin{lemma}[A spherical quadratic penalty]\label{lem:quadratic-penalty}
Let $X_1,\ldots,X_d\in M_n(\C)$ be Hermitian, with $d\ge3$, and put $T_x=\sum_a x_aX_a$. Let $m\ge1$ be an integer. Suppose $B,C>0$ and $u>1$ satisfy
\begin{equation}\label{eq:quadratic-cap-condition}
 \frac{127C}{128B}\ge\frac{u+u^{-1}}2,
 \qquad u^{2m}\ge4.
\end{equation}
Then the matrix
\[
 R_2=\alpha\int\Tch_m(T_x/B)^2\,d\sigma_d(x),
 \qquad \alpha=2^{4d+3}u^{-2m},
\]
is positive semidefinite and satisfies
\[
 \pm T_x\preceq C(I_n+R_2)\qquad(\|x\|_2=1).
\]
\end{lemma}

\begin{proof}
Set
\[
  P(y):=\Tch_m(y/B)^2.
\]
By the properties of the Chebyshev polynomial recorded in
Section~\ref{sec:chebyshev-background}, $P$ is even and nonnegative,
$P(B)=1$, and all its roots lie in $(-B,B)$.  Hence $P$ is increasing
and convex on $[B,\infty)$.  Therefore
\[
  g(y):=
  \begin{cases}
    0,& |y|\le B,\\
    P(y)-1,& |y|>B
  \end{cases}
\]
is even, nonnegative, and convex on $\R$, with $P\ge g$.

We first calibrate the spherical average at the level $C$.  By
\eqref{eq:real-density},
\begin{equation}\label{eq:real-cap-mass}
  \sigma_d
  \left\{
    x:
    \frac{127}{128}
    \le |x_1|
    \le
    \frac{255}{256}
  \right\}
  \ge
  2^{-4d}.
\end{equation}
Indeed, the two intervals have total contribution at least
\[
  2\cdot\frac1{256}\cdot\frac12
  \left(\frac1{256}\right)^{(d-3)/2},
\]
since
\[
  1-\left(\frac{255}{256}\right)^2
  =
  \frac{511}{65536}
  >
  \frac1{256}.
\]
On this cap,
\[
  \frac{|Cx_1|}{B}
  \ge
  \frac{127C}{128B}
  \ge
  \frac{u+u^{-1}}2,
\]
so \eqref{eq:chebyshev-growth} gives
\[
  g(Cx_1)
  \ge
  \frac{u^{2m}+u^{-2m}-2}{4}
  \ge
  \frac{u^{2m}}8,
\]
where the last inequality uses $u^{2m}\ge4$.  Thus, with
\[
  \mathcal J(v)
  :=
  \int g(vx_1)\,d\sigma_d(x),
\]
we have
\begin{equation}\label{eq:J-calibration}
  \alpha\mathcal J(C)\ge1.
\end{equation}
Moreover, $\mathcal J$ is convex and $\mathcal J(0)=0$.

Now fix a unit vector $\psi\in\C^n$ and put
\[
  w_a:=\langle\psi,X_a\psi\rangle\in\R,
  \qquad
  v:=\left(\sum_{a=1}^d w_a^2\right)^{1/2}.
\]
For $x\in S_\R^{d-1}$,
\[
  \langle\psi,T_x\psi\rangle
  =
  \sum_{a=1}^d x_aw_a.
\]
Since $P\ge g$, scalar Jensen's inequality applied to the spectral
measure of $T_x$ gives
\[
  \langle\psi,P(T_x)\psi\rangle
  \ge
  g\!\left(\langle\psi,T_x\psi\rangle\right).
\]
Averaging over $x$ and using rotational invariance of $\sigma_d$,
\begin{equation}\label{eq:R2-J-bound}
  \langle\psi,R_2\psi\rangle
  \ge
  \alpha
  \int
  g\!\left(\sum_{a=1}^d x_aw_a\right)
  d\sigma_d(x)
  =
  \alpha\mathcal J(v).
\end{equation}

If $v>C$, convexity of $\mathcal J$ and $\mathcal J(0)=0$ imply
\[
  \mathcal J(v)
  \ge
  \frac{v}{C}\mathcal J(C).
\]
Combining this with \eqref{eq:J-calibration} and
\eqref{eq:R2-J-bound} yields
\[
  v
  \le
  C\langle\psi,R_2\psi\rangle.
\]
If $v\le C$, trivially
\[
  v
  \le
  C\bigl(1+\langle\psi,R_2\psi\rangle\bigr).
\]
Hence in all cases
\begin{equation}\label{eq:v-R2-bound}
  v
  \le
  C\bigl(1+\langle\psi,R_2\psi\rangle\bigr).
\end{equation}

Finally, for every unit $x\in\R^d$,
\[
  |\langle\psi,T_x\psi\rangle|
  \le
  \|x\|_2\,\|w\|_2
  =
  v.
\]
Together with \eqref{eq:v-R2-bound}, this gives
\[
  |\langle\psi,T_x\psi\rangle|
  \le
  C\langle\psi,(I_n+R_2)\psi\rangle.
\]
Since this holds for every unit $\psi$,
\[
  -C(I_n+R_2)
  \preceq
  T_x
  \preceq
  C(I_n+R_2),
\]
which is the desired conclusion.
\end{proof}

Apply \cref{lem:quadratic-penalty} with $X_a=T_{A_a}$, $B=J_r$, $C=C_r$, and $u=10/9$. The first hypothesis follows from \eqref{eq:target-ratios} and
\begin{equation}\label{eq:cap-rational-check}
 \left(\frac{127}{128}\right)^2\frac{125}{121}
 >\left(\frac{181}{180}\right)^2,
 \qquad \frac{181}{180}=\frac12\left(\frac{10}{9}+\frac9{10}\right).
\end{equation}
After clearing denominators, this is $65322450000>64947503104$. The second follows from \eqref{eq:growth-certificates} and $m\ge64$. We obtain
\begin{equation}\label{eq:R2}
 R_2=\alpha_2\int_{\|A\|_{\HS}=1}
       \Tch_m(T_A/J_r)^2\,d\sigma_d(A),
 \qquad \alpha_2=2^{4d+3}(9/10)^{2m},
\end{equation}
with $\pm T_A\preceq C_r(I_n+R_2)$ for every unit $A\in\mathcal A_D$.

\subsection{A complex spherical penalty for the largest eigenvalue}
To enforce the second compression bound in \eqref{eq:compression-interface}, define
\begin{equation}\label{eq:Rinfty}
 R_\infty=\alpha_\infty\int_{\|z\|_2=1}
       \left(\frac{L_z^*L_z}{32^r}\right)^m\,d\nu_D(z),
 \qquad \alpha_\infty=2^{6(D-1)}\left(\frac{2048}{2079}\right)^m.
\end{equation}
The normalization $32^r$ is the squared linear free-norm bound in \eqref{eq:free-bounds} at $\gamma=8$.

\begin{lemma}[A spherical Gram penalty]\label{lem:Gram-penalty}
For the parameters \eqref{eq:parameter-list}, the matrix $R_\infty$ is positive semidefinite and satisfies
\begin{equation}\label{eq:Gram-order-bound}
 L_z^*L_z\preceq B_r(I_n+R_\infty)
 \qquad(\|z\|_2=1).
\end{equation}
\end{lemma}

\begin{proof}
Positivity of $R_\infty$ is immediate from its definition.  Fix a
unit vector $\psi\in\C^n$ and let
\[
  G_\psi
  :=
  [\langle W_s\psi,W_t\psi\rangle]_{s,t=1}^D
  \succeq0.
\]
Then, for every $y\in\C^D$,
\[
  \langle\psi,L_y^*L_y\psi\rangle
  =
  y^*G_\psi y.
\]
Let $\lambda=\|G_\psi\|$ and let $v$ be a corresponding unit
eigenvector.  By the spectral decomposition,
\[
  y^*G_\psi y
  \ge
  \lambda|\langle v,y\rangle|^2.
\]
Since $L_y^*L_y\succeq0$ and $m\ge1$, scalar Jensen's inequality gives
\[
  \langle\psi,(L_y^*L_y)^m\psi\rangle
  \ge
  \langle\psi,L_y^*L_y\psi\rangle^m.
\]

By \eqref{eq:complex-cap}, the set
\[
  \{y:|\langle v,y\rangle|^2\ge63/64\}
\]
has $\nu_D$-measure $64^{-(D-1)}$.  Restricting the integral in
\eqref{eq:Rinfty} to this cap therefore yields
\[
  \langle\psi,R_\infty\psi\rangle
  \ge
  \left(
    \frac{2048}{2079}
    \frac{63}{64}
    \frac{\lambda}{32^r}
  \right)^m
  =
  \left(
    \frac{32}{33}
    \frac{\lambda}{32^r}
  \right)^m.
\]
By \eqref{eq:target-ratios},
\[
  \frac{B_r}{32^r}\ge\frac{33}{32},
\]
and hence
\[
  \langle\psi,R_\infty\psi\rangle
  \ge
  \left(\frac{\lambda}{B_r}\right)^m.
\]
If $\lambda>B_r$, the right-hand side is at least
$\lambda/B_r$; if $\lambda\le B_r$, the desired estimate is
immediate.  Thus in both cases
\[
  \lambda
  \le
  B_r\bigl(1+\langle\psi,R_\infty\psi\rangle\bigr).
\]
Finally, for every unit $z\in\C^D$,
\[
  \langle\psi,L_z^*L_z\psi\rangle
  =
  z^*G_\psi z
  \le
  \lambda
  \le
  B_r\langle\psi,(I_n+R_\infty)\psi\rangle.
\]
Since this holds for every unit $\psi$, we obtain
\[
  L_z^*L_z\preceq B_r(I_n+R_\infty),
\]
as required.
\end{proof}

\subsection{The common filter and its two uniform bounds}

The two penalties can be combined without a commutation assumption. Set
\begin{equation}\label{eq:filter}
 R=R_2+R_\infty,\qquad F=(I_n+R)^{-1/2}.
\end{equation}
The spherical averages in \eqref{eq:R2} and \eqref{eq:Rinfty} are
explicit finite sums: after polynomial expansion, the scalar
coefficients are evaluated by \eqref{eq:wick-real} and
\eqref{eq:wick-complex}, respectively.

\begin{proposition}[Uniform compression and algebraicity]\label{prop:compression}
For the branch matrices \eqref{eq:finite-branches} and parameters \eqref{eq:parameter-list}, the matrix $R$ is real symmetric and positive semidefinite. The filter $F$ has real algebraic entries and satisfies
\[
 F=F^T,\qquad0\prec F\preceq I_n.
\]
For every $A\in\mathcal A_D$ and $z\in\C^D$,
\begin{equation}\label{eq:uniform-compression}
 \|FT_AF\|\le C_r\|A\|_{\HS},\qquad
 FL_z^*L_zF\preceq B_r\|z\|_2^2I_n.
\end{equation}
Consequently, \eqref{eq:channel} defines a channel $\Phi_F$ with real algebraic entries.
\end{proposition}
\begin{proof}
Both penalties are positive semidefinite. Since the $W_s$ are real,
\[
 \overline{T_A}=T_{\bar A},\qquad
 \overline{L_z^*L_z}=L_{\bar z}^*L_{\bar z}.
\]
Complex conjugation preserves the two coefficient spheres and their measures, so both averages are real symmetric. Their finite expansions have algebraic entries by Section~\ref{sec:moments}. The inverse-square-root fact in Section~\ref{sec:chebyshev-background} therefore gives the same property for $F$, and $R\succeq0$ gives $0\prec F\preceq I_n$.

For unit $A$ and $z$, \cref{lem:quadratic-penalty,lem:Gram-penalty} give
\[
 \pm T_A\preceq C_r(I_n+R_2)\preceq C_r(I_n+R),\qquad
 L_z^*L_z\preceq B_r(I_n+R_\infty)\preceq B_r(I_n+R).
\]
Conjugating by $F$ and using $F(I_n+R)F=I_n$ proves \eqref{eq:uniform-compression}; homogeneity gives arbitrary coefficients. Complete positivity and trace preservation follow from the channel formula as noted in Section~\ref{sec:main}, and its coefficients are real algebraic.
\end{proof}

\subsection{Exact transfer and the average cost}
It remains to verify that the compression has small average loss, as required in \eqref{eq:cost-interface}. For any $R\succeq0$ and $F=(I_n+R)^{-1/2}$,
\begin{equation}\label{eq:filter-loss-bound}
 0\preceq I_n-F^2=R(I_n+R)^{-1}\preceq R.
\end{equation}
Thus it suffices to estimate $\tau_n(R)$ using the finite moments from Section~\ref{sec:finite}.

\begin{proposition}[A small normalized trace]\label{prop:trace-cost}
For the parameters \eqref{eq:parameter-list},
\begin{equation}\label{eq:trace-cost}
 \tau_n(R)\le\alpha_2+\alpha_\infty<2^{1-t},
 \qquad (D+1)\tau_n(I_n-F^2)<\delta.
\end{equation}
\end{proposition}
\begin{proof}
For unit $A$, \eqref{eq:free-bounds} gives $\|T_A^{\mathrm f}\|\le J_r$. The Chebyshev bound \eqref{eq:chebyshev-bound} and \cref{prop:finite-moments} imply
\[
 \tau_n(R_2)
 =\alpha_2\int\tau_G\!\left(\Tch_m(T_A^{\mathrm f}/J_r)^2\right)d\sigma_d(A)
 \le\alpha_2.
\]
Likewise, $\|(L_z^{\mathrm f})^*L_z^{\mathrm f}\|\le32^r$ for unit $z$, and the second identity in \eqref{eq:linear-moment-transfer} gives
\[
 \tau_n(R_\infty)
 =\alpha_\infty\int\tau_G\!\left(
   \left(\frac{(L_z^{\mathrm f})^*L_z^{\mathrm f}}{32^r}\right)^m
   \right)d\nu_D(z)
 \le\alpha_\infty.
\]
Both measures have total mass one. By \eqref{eq:growth-certificates}, $m=64h$, and the two lower bounds defining $h$,
\[
 \alpha_2<2^{4d+3-16h}\le2^{-t},\qquad
 \alpha_\infty<2^{6(D-1)-h}\le2^{-t}.
\]
Finally, \eqref{eq:filter-loss-bound} and $t=s+\log_2D+3$ yield
\[
 (D+1)\tau_n(I_n-F^2)
 \le(D+1)\tau_n(R)
 <(D+1)2^{1-t}
 =\frac{D+1}{4D}\,2^{-s}<\delta.
\]
\end{proof}

\subsection{A quadratic filter for the von Neumann example}\label{sec:small-filter}
The certificate \cref{prop:von-Neumann-certificate} requires only quadratic compression and small average loss. For the branch family \eqref{eq:small-branches}, set
\begin{equation}\label{eq:small-parameters}
 \begin{aligned}
 D_*&=256,&d_*&=256\cdot255=65280,\\
 m_*&=3(4d_*+53)=783519,&\delta_*&=2^{-40},\\
 C_*&=\sqrt{103/2},&J_*&=\sqrt{50}.
 \end{aligned}
\end{equation}
Use $m=m_*$ in \eqref{eq:direct-alphabet} and \eqref{eq:small-branches}, and put $n_*=q_{m_*}^2$. With
\[
 T_{*,A}=\sum_{s\ne t}A_{st}\widetilde W_s^*\widetilde W_t,
 \qquad A\in\mathcal A_{D_*},
\]
define
\begin{equation}\label{eq:small-filter}
 \begin{aligned}
 R_*&=\alpha_*\int_{\|A\|_{\HS}=1}
       \Tch_{m_*}(T_{*,A}/J_*)^2\,d\sigma_{d_*}(A),\\
 \alpha_*&=2^{4d_*+3}(8/9)^{2m_*},\qquad
 F_*=(I_{n_*}+R_*)^{-1/2}.
 \end{aligned}
\end{equation}
The basis \eqref{eq:coefficient-basis} and formula \eqref{eq:wick-real} specify this average as a finite sum. Let $\Phi_*:=\Phi_{F_*}$ in \eqref{eq:channel}, with branches $\widetilde W_s$ and dimensions $n_*,D_*$.

\begin{proposition}[Verification of the smaller filter]\label{prop:small-filter}
For \eqref{eq:small-parameters}, the filter $F_*$ is real symmetric with real algebraic entries and $0\prec F_*\preceq I_{n_*}$. For every $A\in\mathcal A_{D_*}$,
\begin{equation}\label{eq:small-filter-interface}
 \|F_*T_{*,A}F_*\|\le C_*\|A\|_{\HS},
 \qquad (D_*+1)\tau_{n_*}(I_{n_*}-F_*^2)<\delta_*.
\end{equation}
In particular, $\Phi_*$ is a channel with real algebraic entries.
\end{proposition}
\begin{proof}
For $\gamma=2$, $k=64$, and $r=1$, \eqref{eq:free-bounds} gives $J_*^2=16\cdot2^2-7\cdot2=50$. Apply \cref{lem:quadratic-penalty} with $B=J_*$, $C=C_*$, and $u=9/8$. Its first hypothesis follows from
\[
 \left(\frac{127}{128}\right)^2\frac{103}{100}
 -\left(\frac{145}{144}\right)^2
 =\frac{4247}{132710400}>0,
 \qquad \frac{145}{144}=\frac12\left(\frac98+\frac89\right).
\]
Also $(9/8)^6>2$, since $531441>524288$, and $m_*/3=4d_*+53$. Thus $u^{2m_*}>4$. The lemma gives
$\pm T_{*,A}\preceq C_*(I_{n_*}+R_*)$ for unit $A$; conjugation by $F_*$ and homogeneity prove the compression bound.

By \cref{prop:finite-moments} and \eqref{eq:chebyshev-bound},
\begin{equation}\label{eq:small-trace-cost}
 \tau_{n_*}(R_*)\le\alpha_*
 <2^{4d_*+3-m_*/3}=2^{-50}.
\end{equation}
Equation~\eqref{eq:filter-loss-bound}, applied to $R_*$, then gives
\[
 (D_*+1)\tau_{n_*}(I_{n_*}-F_*^2)
 <257\cdot2^{-50}<2^{-40}=\delta_*.
\]
As in the proof of \cref{prop:compression}, invariance of the coefficient sphere under conjugation makes $R_*$ real symmetric. Sections~\ref{sec:moments} and \ref{sec:chebyshev-background} give algebraic entries for $R_*$ and $F_*$. Positivity yields $0\prec F_*\preceq I_{n_*}$, so \eqref{eq:channel} defines the asserted real algebraic channel.
\end{proof}

\section{Proof of the main results}\label{sec:proof}
\begin{proof}[Proof of \Cref{thm:main}]
Fix $p_0>0$ and an integer $r\ge1$. Choose \eqref{eq:parameter-list}, the branch matrices \eqref{eq:finite-branches}, and the filter \eqref{eq:filter}, and write $\Phi_{r,p_0}:=\Phi_F$.
\cref{prop:finite-moments} verifies \eqref{eq:Bell-moment-interface}, and \cref{prop:compression} gives real algebraicity and the compression bounds with $C=C_r$, $B=B_r$. \cref{prop:trace-cost} and \eqref{eq:error-budget} verify \eqref{eq:cost-interface}. Hence \Cref{thm:certificate} gives
\[
 \inf_{p\in[p_0,\infty]}\Delta_p(\Phi_{r,p_0};\Omega_n)
 \ge\frac r2+\frac{16r-1}{16384}>\frac r2.
\]
The dimensions are $n=q_{3m}^{8r}=2^{N(r,p_0)}$ and $D=2^{15r}$ by \eqref{eq:finite-branches} and \eqref{eq:input-budget}; the bound \eqref{eq:main-dimension-bound} was verified in Section~\ref{sec:filter-parameters}. All matrices and the Bell input depend only on $p_0$ and $r$, not on the running order $p$. Finally, $r=\lceil2g\rceil$ gives $r/2\ge g$, proving the prescribed-gap assertion.
\end{proof}

\begin{proof}[Proof of \cref{cor:von-neumann}]
Take the branch matrices \eqref{eq:small-branches} and filter \eqref{eq:small-filter} with parameters \eqref{eq:small-parameters}. Propositions~\ref{prop:finite-moments} and \ref{prop:small-filter} verify the hypotheses of \cref{prop:von-Neumann-certificate} with
\[
 \gamma=2,\qquad k=64,\qquad K=256,\qquad
 C^2=103/2,\qquad\delta=2^{-40}<K^{-2}.
\]
Thus $\varepsilon=K^2\delta=2^{-24}$ and $1+C^2/K=615/512$, and the certificate gives
\begin{equation}\label{eq:small-gap-bound}
 \Delta_1(\Phi_*;\Omega_{n_*})
 \ge\frac{17}{32}\log2-2\log(615/512)
      -16\varepsilon\log2-h_{\mathrm b}(\varepsilon).
\end{equation}
We bound the fixed margin and the perturbation term separately. For $z=103/1127$, the series \eqref{eq:log-series} gives
\[
 \log(615/512)
 \le2\left(z+\frac{z^3}{3}+\frac{z^5}{5(1-z^2)}\right)
 <\frac{183299}{1000000}.
\]
Together with \eqref{eq:log-certificates}, this yields
\[
 \frac{17}{32}\log2-2\log(615/512)
 >\frac{12891}{8000000}>\frac1{640}.
\]
The same series, keeping its first term at $z=1/3$ and bounding the tail geometrically, gives $\log2<25/36<7/10$. Since
$h_{\mathrm b}(\varepsilon)\le\varepsilon(1-\log\varepsilon)$,
\[
 16\varepsilon\log2+h_{\mathrm b}(\varepsilon)
 \le(40\log2+1)2^{-24}<29\cdot2^{-24}.
\]
Substituting into \eqref{eq:small-gap-bound} proves
\[
 \Delta_1(\Phi_*;\Omega_{n_*})
 >\frac1{640}-\frac{29}{2^{24}}>\frac1{1024}.
\]
Finally, \eqref{eq:leg-dimension} and \eqref{eq:small-parameters} give
\[
 n_*=q_{m_*}^2=2^{48m_*-14}=2^{37608898},\qquad D_*=256=2^8,
\]
which proves the stated dimensions and \eqref{eq:small-gap}.
\end{proof}

\section{Comparison of Finite-Dimensional Realizations}\label{sec:discussion}

The construction of Shou and Gorshkov also gives arbitrarily large uniform gaps on every prescribed $[p_0,\infty]$, with $p_0>0$ \cite[Theorem~1.3]{ShouGorshkov2026}. Both constructions combine finite moment matching with damping, rather than strong convergence. Their finite model uses Sanov generators reduced modulo a power of two and
the left regular representation of a product of
$\mathrm{SL}_2(\Z/2^a\Z)$ groups
\cite[Section~B.2]{ShouGorshkov2026}.
In Section~\ref{sec:finite}, we instead use permutation actions on
explicitly indexed invariant fibres of high-girth graphs over the fixed
field $\mathbb F_{256}$. Nontrivial words within the prescribed length
range have no fixed points, giving the required exact trace identities.
The filters also differ:
the lattice-net sums of even powers in \cite[Eq.~(B.21)]{ShouGorshkov2026}
are replaced here by spherical averages of Chebyshev and power penalties,
evaluated by finite moment formulas in Section~\ref{sec:filter}.
Thus the comparison below is between the stated quantitative bounds,
rather than between the best possible dimensions attainable by the two
finite-realization paradigms.

For a common target gap $g>0$ and cutoff $p_0>0$, set
\[
 r=\lceil2g\rceil,\qquad
 \ell=\left\lceil\frac54\left(6+\frac{g}{\log2}\right)\right\rceil,
 \qquad \theta=\min\{p_0,1/2\}.
\]
Then \Cref{thm:main} and \cite[Theorem~1.3]{ShouGorshkov2026}
both certify a gap greater than $g$ simultaneously for all
$p\in[p_0,\infty]$, using a Bell input in each construction.
For a channel with Hilbert-space dimensions
$d_{\mathrm{in}},d_{\mathrm{out}}$,
write $B_{\mathrm{in}}=\lceil\log_2d_{\mathrm{in}}\rceil$ and
$B_{\mathrm{out}}=\lceil\log_2d_{\mathrm{out}}\rceil$.
The respective prescriptions give
\begin{equation}\label{eq:discussion-budgets}
 \begin{aligned}
 (B_{\mathrm{in}}^{\mathrm{ours}},
  B_{\mathrm{out}}^{\mathrm{ours}})
   &=(N(r,p_0),15r),\\
 (B_{\mathrm{in}}^{\mathrm{SG}},
  B_{\mathrm{out}}^{\mathrm{SG}})
   &=\left(
       \left\lceil16\ell
       (336m_{\mathrm{SG}}-2+\log_2 3)\right\rceil,
       17\ell
     \right),
 \end{aligned}
\end{equation}
where $N(r,p_0)$ is given by \eqref{eq:input-budget} and
\begin{equation}\label{eq:discussion-sg-moment}
 m_{\mathrm{SG}}
 =4+\left\lceil
       \max\left\{\frac{\ell+7}{\theta},17\ell+7\right\}
     \right\rceil
    +(17\ell+5)2^{34\ell-1}.
\end{equation}
These SG parameters are those of \cite[Eq.~(1.5)]{ShouGorshkov2026},
with gap bound $(4\ell/5-6)\log2$.
For illustration, we take $p_0=1/4$.  This lies in the regime where the dependence on $\theta=\min\{p_0,1/2\}$ is nontrivial; the choice is otherwise inessential. The qubit counts are as follows. Input counts are rounded to five
significant figures; output counts are exact.
\begin{center}
\small
\begin{tabular}{@{}ccccc@{}}
\toprule
& \multicolumn{2}{c}{Present construction}
& \multicolumn{2}{c}{Shou--Gorshkov, Theorem~1.3}\\
Target $g$ (nats)
& $B_{\mathrm{in}}$ & $B_{\mathrm{out}}$
& $B_{\mathrm{in}}$ & $B_{\mathrm{out}}$\\
\midrule
$1/2$ & $9.8953\times10^{12}$ & $15$
      & $4.9832\times10^{98}$ & $153$\\
$1$   & $2.1251\times10^{22}$ & $30$
      & $1.0536\times10^{109}$ & $170$\\
$2$   & $4.9001\times10^{40}$ & $60$
      & $4.4565\times10^{129}$ & $204$\\
\bottomrule
\end{tabular}
\end{center}
These are sufficient dimensions from the stated parameter choices,
not lower bounds on the dimensions needed by either method.

More generally, $15r<17\ell$ for every $g>0$.
For fixed $p_0>0$, the exact prescriptions also give, as $g\to\infty$,
\[
 B_{\mathrm{in}}^{\mathrm{ours}}
   =\Theta\!\left(g2^{60g}\right),
 \qquad
 B_{\mathrm{in}}^{\mathrm{SG}}
   =\Theta\!\left(g^2 2^{85g/(2\log2)}\right).
\]
Since $85/(2\log2)>60$, this improves the large-gap growth rate of
the explicit input-qubit bound as well as the output-qubit budget.
This asymptotic comparison concerns fixed $p_0>0$ and should not be
interpreted as uniform input-dimension superiority under arbitrary joint
scalings of $g$ and $p_0$, especially since neither construction is
dimensionally optimized.

\section{Statements and Declarations}
\label{sec:statements-declarations}

\subsection{Use of Artificial Intelligence}

Inspired by the recent constructions of Shou and Gorshkov \cite{ShouGorshkov2026} and Cheng and Wu \cite{ChengWu2026}, the authors identified the core technical route for obtaining a finite-dimensional closed-form counterexample: replace the strong-convergence input normally used to transfer free-group estimates to finite dimensions by finite-order moment matching, together with a deterministic filter that suppresses exceptional directions while preserving the two-copy correlations needed for the Bell-state witness.

The authors then presented this idea to ChatGPT 6 Astra and asked whether it could be combined with the direct-product free-group model of \cite{ZhenEtAlAllOrders2026} to produce a finite-dimensional closed-form counterexample. Then the LLM developed a technical route resolving this problem, including the finite-dimensional realization, filtering construction, and entropy estimates.  The authors subsequently worked through the argument interactively with the LLM, carefully checking and assimilating the full technical chain and organizing the manuscript according to their own mathematical logic. All mathematical arguments, interpretations, references, and final editorial decisions were independently checked and approved by the authors, who take full responsibility for the scientific content of the paper.

\subsection{Acknowledgements}

We are grateful to Prof.~Peixue Wu for promptly sharing the recent works \cite{ShouGorshkov2026,ChengWu2026}, which allowed us to follow the latest developments at an early stage and helped us identify the core technical route that led to the present construction. This work was supported by the National Key R\&D Program of China (Grant No.~2024YFE0102500), the National Natural Science Foundation of China (Grant. No.~92576114, 12447107), the Guangdong Provincial Quantum Science Strategic Initiative (Grant No.~GDZX2403008, GDZX2503001), and the Guangdong Provincial Key Lab of Integrated Communication, Sensing and Computation for Ubiquitous Internet of Things (Grant No.~2023B1212010007).

\end{document}